\documentclass[preprint,3p,times]{elsarticle}

\usepackage{amssymb}
\usepackage{amsmath}
\usepackage[all,cmtip]{xy}
\usepackage{amsthm}
\usepackage{color}
\usepackage{enumitem}
\usepackage{tikz}
\usepackage[colorlinks=true]{hyperref}
\usepackage[hyphenbreaks]{breakurl}
\input diagxy

\theoremstyle{plain}
  \newtheorem{thm}{Theorem}[section]
  \newtheorem{lem}[thm]{Lemma}
  \newtheorem{prop}[thm]{Proposition}
  
\theoremstyle{definition}
  \newtheorem{defn}[thm]{Definition}

  \newtheorem{exmps}[thm]{Examples}
  \newtheorem{rem}[thm]{Remark}

\DeclareMathAlphabet{\mathcal}{OMS}{cmsy}{m}{n}

\DeclareMathOperator{\ub}{ub}
\DeclareMathOperator{\lb}{lb}
\DeclareMathOperator{\id}{id}

\DeclareMathOperator{\supp}{supp}

\makeatletter
\def\ps@pprintTitle{%
 \let\@oddhead\@empty
  \let\@evenhead\@empty
  \def\@oddfoot{\vbox{\hsize=\textwidth\footnotesize
  \vskip 8pt
  \copyright 2018. This manuscript version is made available under the CC-BY-NC-ND 4.0 license \url{https://creativecommons.org/licenses/by-nc-nd/4.0/}. The published version is available at \url{https://doi.org/10.1016/j.fss.2018.05.016}.\\
  }}%
  \let\@evenfoot\@oddfoot}
\makeatother

\def\oto{{\bfig\morphism<180,0>[\mkern-4mu`\mkern-4mu;]\place(86,0)[\circ]\efig}}
\def\rto{{\bfig\morphism<180,0>[\mkern-4mu`\mkern-4mu;]\place(82,0)[\mapstochar]\efig}}
\def\nra{\relbar\joinrel\joinrel\mapstochar\joinrel\joinrel\rightarrow}
\def\nla{\leftarrow\joinrel\joinrel\joinrel\mapstochar\joinrel\relbar}

\newcommand{\da}{\downarrow}
\newcommand{\ua}{\uparrow}
\newcommand{\ra}{\rightarrow}
\newcommand{\la}{\leftarrow}

\newcommand{\lra}{\longrightarrow}
\newcommand{\lda}{\swarrow}
\newcommand{\rda}{\searrow}

\newcommand{\rat}{\rightarrowtail}
\newcommand{\bv}{\bigvee}
\newcommand{\bw}{\bigwedge}

\newcommand{\dv}{\dashv}

\newcommand{\nat}{\natural}

\renewcommand{\phi}{\varphi}
\newcommand{\al}{\alpha}
\newcommand{\be}{\beta}

\newcommand{\Lam}{\Lambda}
\newcommand{\lam}{\lambda}

\newcommand{\Om}{\Omega}

\newcommand{\CC}{\mathcal{C}}
\newcommand{\CD}{\mathcal{D}}

\newcommand{\CM}{\mathcal{M}}

\newcommand{\sP}{{\sf P}}

\newcommand{\sy}{{\sf y}}

\newcommand{\FQ}{\mathfrak{Q}}

\newcommand{\Ord}{{\bf Ord}}

\newcommand{\Rel}{{\bf Rel}}
\newcommand{\Set}{{\bf Set}}
\newcommand{\Sup}{{\bf Sup}}

\newcommand{\QOrd}{\FQ\text{-}\Ord}

\newcommand{\dphi}{\phi^{\da}}
\newcommand{\uphi}{\phi_{\ua}}

\newcommand{\dpsi}{\psi^{\da}}
\newcommand{\upsi}{\psi_{\ua}}
\newcommand{\hphi}{\widehat{\phi}}
\newcommand{\tphi}{\widetilde{\phi}}

\newcommand{\sPd}{\sP^{\dag}}
\newcommand{\syd}{\sy^{\dag}}

\newcommand{\co}{{\rm co}}
\newcommand{\op}{{\rm op}}
\newcommand{\PX}{\sP X}
\newcommand{\PY}{\sP Y}
\newcommand{\PZ}{\sP Z}
\newcommand{\PdX}{\sPd X}
\newcommand{\PdY}{\sPd Y}

\newcommand{\DQ}{\CD\FQ}
\newcommand{\FDist}{{\bf FDist}}
\newcommand{\FOrd}{{\bf FOrd}}
\newcommand{\FRel}{{\bf FRel}}
\newcommand{\FSup}{{\bf FSup}}
\newcommand{\POrd}{{\bf POrd}}
\newcommand{\QFDist}{\FQ\text{-}\FDist}
\newcommand{\QFOrd}{\FQ\text{-}\FOrd}
\newcommand{\QFRel}{\FQ\text{-}\FRel}
\newcommand{\QFSup}{\FQ\text{-}\FSup}
\newcommand{\QPOrd}{\FQ\text{-}\POrd}

\newcommand{\ldd}{\mathrel{\slash}}
\newcommand{\rdd}{\mathrel{\backslash}}
\newcommand{\with}{\mathrel{\&}}

\newcommand{\comp}{\mathrel{\circ}}

\numberwithin{equation}{section}

\begin{document}

\begin{frontmatter}



\title{Fuzzy Galois connections on fuzzy sets}


\author[B]{Javier~Guti\'{e}rrez~Garc\'{\i}a}
\ead{javier.gutierrezgarcia@ehu.eus}

\author[S]{Hongliang~Lai}
\ead{hllai@scu.edu.cn}

\author[S]{Lili~Shen\corref{cor}}
\ead{shenlili@scu.edu.cn}

\cortext[cor]{Corresponding author.}
\address[B]{Department of Mathematics, University of the Basque Country UPV/EHU, Apdo. 644, 48080 Bilbao, Spain}
\address[S]{School of Mathematics, Sichuan University, Chengdu 610064, China}

\begin{abstract}
In fairly elementary terms this paper presents how the theory of preordered fuzzy sets, more precisely quantale-valued preorders on quantale-valued fuzzy sets, is established under the guidance of enriched category theory. Motivated by several key results from the theory of quantaloid-enriched categories, this paper develops all needed ingredients purely in order-theoretic languages for the readership of fuzzy set theorists, with particular attention paid to fuzzy Galois connections between preordered fuzzy sets.
\end{abstract}


\end{frontmatter}




\section{Introduction}

The theory of fuzzy preorders was initiated by Zadeh's pioneering work \cite{Zadeh1971} and has been developed for decades, during which time the table of truth-values under concern has been extended from the unit interval $[0,1]$ to a unital quantale $\FQ$ \cite{Rosenthal1990}. With the multiplication $\with\colon\FQ\times\FQ\to\FQ$ of a unital quantale $\FQ$ playing the role of the logical conjunction and its unit $e$ representing the logical value ``true'', a $\FQ$-preorder on a set $X$ is given by a map ${\al\colon X\times X\to\FQ}$ such that
\begin{equation} \label{Q-preorder-crisp}
e\leq\al(x,x)\ \ \text{(reflexivity)}\quad\text{and}\quad\al(y,z)\with\al(x,y)\leq\al(x,z)\ \ \text{(transitivity)}
\end{equation}
for all $x,y,z\in X$; here the transitivity condition is also formulated by some authors as $\al(x,y)\with\al(y,z)\leq\al(x,z)$ (see, e.g., \cite{Hoehle2015,Stubbe2014}), which in fact defines $\FQ^{\tau}$-preorders on $X$ in the sense of \eqref{Q-preorder-crisp}, with $\FQ^{\tau}$ being the \emph{conjugate} of the quantale $\FQ$ (see Remark~\ref{Q-valued-preorder}). $\FQ$-preordered sets have attracted wide attention in the fuzzy community; see \cite{Bvelohlavek2002,Bvelohlavek2004,Denniston2014,Hoehle2016,Hoehle1985,Lai2006,Lai2009,Shen2013,Yao2009} for instance.

While $\FQ$-preordered sets defined by \eqref{Q-preorder-crisp} are actually $\FQ$-preorders on \emph{crisp} sets, recently Lai and Zhang and their co-authors have established the theory of $\FQ$-preorders on \emph{fuzzy} sets especially when $\FQ$ is a \emph{divisible} quantale \cite{Li2017,Pu2012,Tao2014}; similar approaches have been adopted by H{\"o}hle and Kubiak for the construction of their quantale-valued preorders \cite{Hoehle2015,Hoehle2011a}. The key machinery involved in these works is that of categories enriched in a \emph{quantaloid} \cite{Rosenthal1996,Stubbe2005,Stubbe2006,Stubbe2014}, which is a special case of categories enriched in a bicategory \cite{Betti1982,Betti1983,Walters1981}. To be specific, each unital quantale $\FQ$ gives rise to a quantaloid $\DQ$ of \emph{diagonals} in $\FQ$ \cite{Hoehle2011a,Pu2012,Stubbe2014}, and a $\FQ$-subset (i.e., a $\FQ$-valued fuzzy set) equipped with a $\FQ$-preorder is exactly a category enriched in the quantaloid $\DQ$.

The purpose of this paper is to present the theory of preordered fuzzy sets, more precisely \emph{$\FQ$-preordered} {\protect\linebreak}\emph{$\FQ$-subsets}, in the most accessible terms for readers from the fuzzy community who may not be familiar with the arsenal of category theorists and, in particular, the theory of quantaloid-enriched categories. For the most generality we only assume $\FQ$ to be a unital quantale, not necessarily commutative, without imposing any divisibility condition as in \cite{Pu2012,Tao2014}, and our focus will be on \emph{$\FQ$-Galois connections} between $\FQ$-preordered $\FQ$-subsets. Although some of the results in this paper are generalizations of those in \cite{GutierrezGarcia2010} for $\FQ$-preordered (crisp) sets, the method developed here, as prepared in Section~\ref{Fuzzy_Relations}, allows for a much wider range of applicability.

To understand the structure of $\FQ$-preordered $\FQ$-subsets, let us first look at their crisp version; that is, when $\FQ=\boldsymbol{2}$, the two-element Boolean algebra. In this case, a $\boldsymbol{2}$-preordered $\boldsymbol{2}$-subset is a \emph{partially defined preordered set} given by a (crisp) set $X$, a (crisp) subset $A\subseteq X$ and a preorder ``$\leq$'' on $A$. Explicitly, for all $x,y,z\in X$:
\begin{enumerate}[label=(P\arabic*)]
\item (divisibility) $x\leq y$ only if $x,y\in A$;
\item (reflexivity) if $x\in A$, then $x\leq x$;
\item (transitivity) if there exists $y\in A$ such that $y\leq z$ and $x\leq y$, then $x\leq z$.
\end{enumerate}
Now, replacing $\boldsymbol{2}$ with a general unital quantale $\FQ$, a $\FQ$-subset consists of a (crisp) set $X$ and a map $|\text{-}|\colon X\to\FQ$, and a $\FQ$-preorder on $(X,|\text{-}|)$ is given by a map $\al\colon X\times X\to\FQ$ satisfying
\begin{enumerate}[label=($\FQ$P\arabic*),leftmargin=3.2em]
\item \label{Qd} (divisibility) $(\al(x,y)\ldd |x|)\with |x|=\al(x,y)=|y|\with(|y|\rdd\al(x,y))$,
\item \label{Qr} (reflexivity) $|x|\leq\al(x,x)$,
\item \label{Qt} (transitivity) $(\al(y,z)\ldd |y|)\with\al(x,y)=\al(y,z)\with(|y|\rdd\al(x,y))\leq\al(x,z)$
\end{enumerate}
for all $x,y,z\in X$, where $\ldd$, $\rdd$ stand for the left and the right implications in $\FQ$. Hence, while the notion of $\FQ$-preordered set defined by \eqref{Q-preorder-crisp} extend the notion of ``preordered set'', the notion of \emph{$\FQ$-preordered $\FQ$-subset} defined as above is actually a generalization of the notion of ``partially defined preordered set''.

However, instead of establishing the theory of $\FQ$-preordered $\FQ$-subsets upon the complicated pointwise definition \ref{Qd}--\ref{Qt} as in most of the literature for $\FQ$-preordered sets, we would rather introduce \emph{$\FQ$-relations} between {\protect\linebreak}$\FQ$-subsets (i.e., \emph{fuzzy relations} between fuzzy sets), in Section~\ref{Fuzzy_Relations}, as the cornerstone of our theory. Indeed, as it is well known that (crisp) preorders are reflexive and transitive (crisp) relations, an appropriate notion of $\FQ$-relation between $\FQ$-subsets (see Definition~\ref{Q-relation}) allows us to define $\FQ$-preorders on $\FQ$-subsets simply as reflexive and transitive $\FQ$-relations (see Definition~\ref{Q-preorder}). More importantly, the calculus of $\FQ$-relations based on the fact that $\FQ$-subsets and $\FQ$-relations constitute a quantaloid significantly simplifies the treatment of $\FQ$-preorders and also makes the related concepts much more elegant.

With necessary discussions of the basic concepts of $\FQ$-preordered $\FQ$-subsets in Section~\ref{Preordered-fuzzy-sets}, we put our emphasis on $\FQ$-Galois connections in Section~\ref{Fuzzy-Galois-connections}. The notion of Galois connections between preordered sets \cite{Birkhoff1967,Herrlich1990,Ore1944} has been extended to the fuzzy setting by B{\v e}lohl{\'a}vek since 1999 \cite{Bvelohlavek1999,Bvelohlavek2002,Bvelohlavek2004}, and in the subsequent works \cite{Georgescu2003,Georgescu2004,GutierrezGarcia2010} of other authors fuzzy Galois connections have been considered in a non-commutative world. More precisely, B{\v e}lohl{\'a}vek's fuzzy Galois connections are $\FQ$-Galois connections between $\FQ$-preordered (crisp) sets, whose prototypes are adjoint functors between quantale-enriched categories \cite{GutierrezGarcia2010,Kelly1982,Lai2007,Wagner1994}. In this paper, based on the notion of adjoint functor between \emph{quantaloid}-enriched categories \cite{Lai2017,Stubbe2005,Stubbe2014}, we extend the realm of fuzzy Galois connections further to {\protect\linebreak}$\FQ$-Galois connections between $\FQ$-preordered $\FQ$-subsets. Motivated by several key results from the theory of quantaloid-enriched categories \cite{Shen2014,Shen2013a,Stubbe2005,Stubbe2006}, we carefully exhibit the interactions of $\FQ$-Galois connections with
\begin{itemize}
\item the completeness of $\FQ$-preordered $\FQ$-subsets,
\item the preservation of suprema, infima and (co)tensors, and
\item $\FQ$-distributors between $\FQ$-preordered $\FQ$-subsets (i.e., $\FQ$-relations that are compatible with the $\FQ$-preorder structures).
\end{itemize}
In particular, we propose a conceptual definition of \emph{$\FQ$-polarities} and (dual) \emph{$\FQ$-axialities}, following the terminologies in \cite{Birkhoff1967,Erne1993a,GutierrezGarcia2010}, as $\FQ$-Galois connections between (dual) $\FQ$-powersets of $\FQ$-preordered $\FQ$-subsets, and their bijective correspondences with $\FQ$-distributors are established.

Without assuming any a-priori background by the readers on quantaloid-enriched categories, this paper is intended to develop all needed ingredients purely in order-theoretic languages, though implicitly under the guidance of enriched category theory and occasionally with remarks pointing out their pivotal links to the categorical concepts. As shall be seen, the implementation of the $\FQ$-relational calculus, our key method that was not usually adopted in the literature, not only presents the theory of $\FQ$-preordered $\FQ$-subsets in a succinct way, but also unveils the conceptual nature of the related notions.

\section{The calculus of fuzzy relations} \label{Fuzzy_Relations}

A \emph{unital quantale} is a triple $(\FQ,\with,e)$ consisting of a complete lattice $\FQ$, an element $e\in\FQ$ and a binary operation $\with$ on $\FQ$, such that
\begin{enumerate}[label={\rm(\arabic*)}]
\item $(\FQ,\with,e)$ is a monoid with $e$ being the unit;
\item $p\with(\bv\limits_{i\in I}q_i)=\bv\limits_{i\in I}(p\with q_i)$ and $(\bv\limits_{i\in I}p_i)\with q=\bv\limits_{i\in I}(p_i\with q)$ for all $p,p_i,q,q_i\in\FQ$ $(i\in I)$.
\end{enumerate}
The corresponding Galois connections
$\bfig
\morphism/@{->}@<4pt>/<500,0>[\FQ`\FQ;-\with q]
\morphism(500,0)|b|/@{->}@<4pt>/<-500,0>[\FQ`\FQ;-\ldd q]
\place(250,5)[\mbox{\footnotesize{$\bot$}}]
\efig$ and $\bfig
\morphism/@{->}@<4pt>/<500,0>[\FQ`\FQ;p\with -]
\morphism(500,0)|b|/@{->}@<4pt>/<-500,0>[\FQ`\FQ;p\rdd -]
\place(250,5)[\mbox{\footnotesize{$\bot$}}]
\efig$ induced by the monoid multiplications satisfy
$$p\with q\leq r\iff p\leq r\ldd q\iff q\leq p\rdd r$$
for all $p,q,r\in\FQ$, where the operations $\ldd$, $\rdd$ are called \emph{left} and \emph{right implications} in $\FQ$, respectively.

Throughout this paper, we let $\FQ=(\FQ,\with,e)$ be a \emph{non-trivial} unital quantale; that is, the bottom element $\bot<e$ in $\FQ$. We say that $\FQ$ is \emph{integral} if $e=\top$, the top element of $\FQ$. $\FQ$ is \emph{commutative} if $p\with q=q\with p$ for all $p,q\in\FQ$, in which case we write $p\ra q$ for $q\ldd p=p\rdd q$.

Taking $\FQ$ as the table of truth-values, a \emph{$\FQ$-subset} (or, \emph{fuzzy set}) is a pair $(X,|\text{-}|_X)$, where $X$ is a crisp set and
$$|\text{-}|_X\colon X\to\FQ$$
is a map, with the value $|x|_X$ interpreted as the \emph{membership degree} of each $x$ in $X$. For the simplicity of notations, in the following we just write $|\text{-}|$ for $|\text{-}|_X$ and $X$ for a $\FQ$-subset $(X,|\text{-}|)$ if no confusion arises, which is always assumed to be equipped with a membership map $|\text{-}|\colon X\to\FQ$. The slice category $\Set/\FQ$ has $\FQ$-subsets as objects, and \emph{membership-preserving} maps $f\colon X\to Y$ between $\FQ$-subsets, i.e., maps $f\colon X\to Y$ with
$$|x|=|fx|$$
for all $x\in X$, as morphisms.

Given an element $q\in\FQ$, following the terminologies in \cite{Hoehle2015}, we say that
\begin{enumerate}[label={\rm(\arabic*)}]
\item $u\in\FQ$ is \emph{left-divisible by $q$} if there exists $p\in\FQ$ such that $q\with p=u$ or, equivalently, if $q\with(q\rdd u)=u$;
\item $u\in\FQ$ is \emph{right-divisible by $q$} if there exists $p\in\FQ$ such that $p\with q=u$ or, equivalently, if $(u\ldd q)\with q=u$.
\end{enumerate}
For any $p,q\in\FQ$, we denote by
$$\DQ(p,q)=\{u\in\FQ\mid (u\ldd p)\with p=u=q\with(q\rdd u)\}$$
the set of elements in $\FQ$ that are simultaneously right-divisible by $p$ and left-divisible by $q$. The quantale $\FQ$ is \emph{divisible} if, whenever $u\le q$ in $\FQ$, $u$ is both left- and right-divisible by $q$, i.e., $q\with(q\rdd u)=u=(u\ldd q)\with q$. It is easy to observe the following facts:

\begin{lem} \label{DQ-elements}
Let $(\FQ,\with,e)$ be a unital quantale and $p,q\in\FQ$. Then{\textup:}
\begin{enumerate}[label={\rm(\arabic*)}]
\item \label{DQ-elements:bot-q}
    $\DQ(\bot,q)=\DQ(q,\bot)=\{\bot\}$.
\item \label{DQ-elements:e-e}
    $\DQ(e,e)=\FQ$.
\item \label{DQ-elements:q-q}
    $q\in\DQ(q,q)$.
\item \label{DQ-elements:bot-p-q}
    $\bot\in\DQ(p,q)$.
\item \label{DQ-elements:e-top}
   $e\in\DQ(\top,\top)$ if, and only if, $\FQ$ is integral.
\end{enumerate}
Moreover,
\begin{enumerate}[label={\rm(\arabic*)},start=6]
\item \label{DQ-elements:integral}
    $\FQ$ is integral if, and only if, $\DQ(p,q)\subseteq \{u\in\FQ\mid u\le p\wedge q\}$ for all $p,q\in\FQ$.
\item \label{DQ-elements:divisible}
    $\FQ$ is divisible if, and only if, $\DQ(p,q)= \{u\in\FQ\mid u\le p\wedge q\}$ for all $p,q\in\FQ$.
\end{enumerate}
\end{lem}

\begin{exmps} \label{Q-exmp} \
\begin{enumerate}[label=(\arabic*)]
\item \label{Q-exmp:c-d}
    (Commutative and divisible quantales) Every frame is a divisible, commutative and idempotent quantale, and vice versa; in particular, so is $\FQ=\boldsymbol{2}$, the two-element Boolean algebra.

    A binary operation $\with$ on the unit interval $[0,1]$ defines a continuous (resp. left-continuous) t-norm on $[0,1]$ if, and only if, $([0,1],\with,1)$ is a commutative and divisible (resp. integral) quantale. In particular, $[0,1]$ equipped with the minimum, the product, or the {\L}ukasiewicz t-norm is a commutative and divisible quantale.

    Lawvere's quantale \cite{Lawvere1973} $([0,\infty]^{\op},+,0)$, where $[0,\infty]^{\op}$ is the extended non-negative real line equipped with the order ``$\geq$'', is commutative and divisible, in which implications are given by
    $$p\ra q=\max\{0,q-p\}.$$
    Indeed, Lawvere's quantale is isomorphic to the quantale $[0,1]$ equipped with the product t-norm.
\item \label{Q-exmp:c-ni}
    (Commutative and non-integral quantales) On the three-chain $C_3=\{\bot,e,\top\}$ we have the commutative unital quantale $(C_3,\with,e)$, with
    $$\top\with\top=\top\ra\top=\top,\quad\top\ra\bot=\top\ra e=\bot$$
    and the other multiplications\,/\,implications being trivial. It follows immediately from Lemma~\ref{DQ-elements} that
    \begin{align*}
    &\CD C_3(\bot,\bot)=\CD C_3(\bot,e)=\CD C_3(\bot,\top)=\CD C_3(e,\bot)=\CD C_3(\top,\bot)=\{\bot\},\\
    &\CD C_3(\top,\top)=\CD C_3(e,\top)=\CD C_3(\top,e)=\{\bot,\top\}\quad\text{and}\quad\CD C_3(e,e)=\{\bot,e,\top\}.
    \end{align*}
    In fact, this quantale is universal among non-integral unital quantales in the sense that in any non-integral unital quantale $\FQ$, the elements $\{\bot,e,\top\}\subseteq\FQ$ form a subquantale that is isomorphic to $(C_3,\with,e)$. For instance, one may embed $(C_3,\with,e)$ into the quantale $\FQ=([0,\infty],\cdot,1)$, whose underlying complete lattice is $[0,\infty]$ equipped with the usual order ``$\leq$'', and whose multiplication is given by the multiplication ``$\cdot$'' of real numbers (under the assumption $0\cdot\infty=0$), with implications given by
    $$p\ra q=\frac{q}{p}.$$
    $([0,\infty],\cdot,1)$ is obviously a commutative and non-integral quantale, and it is not difficult to see that
    $$\DQ(p,q)=\begin{cases}
    \{0\}& \text{if}\ p\wedge q=0,\\
    [0,\infty]& \text{if}\ 0<p, q<\infty,\\
    \{0,\infty\}& \text{if}\ p\vee q=\infty.
    \end{cases}$$
\item \label{Q-exmp:nc-ni}
    (Non-commutative and non-integral quantales) For each complete lattice $L$ with at least two elements, the complete lattice $\Sup(L,L)$ of all $\sup$-preserving maps on $L$ carries a non-trivial and non-commutative quantale structure $(\Sup(L,L),\cdot,1_L)$, where $\cdot$ is the composition of maps, and $1_L$ is the identity map on $L$. Moreover, $(\Sup(L,L),\cdot,1_L)$ is non-integral if, and only if, $L$ contains at least three elements.

    Each non-empty set $X$ gives rise to a non-trivial and non-commutative quantale $(\Rel(X,X),\circ,\id_X)$, where $\Rel(X,X)={\bf 2}^{X\times X}$ is the complete lattice of all relations on $X$, $\circ$ is the composition of relations, and $\id_X=\{(x,x)\mid x\in X\}$ is the identity relation on $X$. Moreover, $\Rel(X,X)$ is non-integral if, and only if, $X$ contains at least two elements.

    Each monoid $M=(M,\with,e)$ induces a \emph{free quantale} $\FQ=({\bf 2}^M,\with,\{e\})$ with $A\with B=\{a\with b\mid a\in A,\,b\in B\}$ for all $A,B\subseteq M$. Clearly, $\FQ$ is non-commutative if and only if so is $M$, while $\FQ$ is non-integral if and only if $M$ contains at least two elements.
\item \label{Q-exmp:nc-nd-i}
    (Non-commutative and non-divisible integral quantales) For each complete chain $L$ with at least three elements, the subset
    $$\Sup(L,L)_{\leq 1_L}=\{f\in\Sup(L,L)\mid f\leq 1_L\}$$
    forms a subquantale of $(\Sup(L,L),\cdot,1_L)$ that is non-commutative, non-divisible and integral. Indeed, in order to see that $(\Sup(L,L)_{\leq 1_L},\cdot,1_L)$ is non-divisible, let $x_0\in L$ with $\bot<x_0<\top$, and let $f,g\in\Sup(L,L)_{\leq 1_L}$ be given by
    $$fx=\begin{cases}
    \bot & \text{if}\ x\leq x_0,\\
    x_0 & \text{if}\ x>x_0
    \end{cases}\quad\text{and}\quad gx=\begin{cases}
    \bot & \text{if}\ x\leq x_0,\\
    x & \text{if}\ x>x_0
    \end{cases}$$
    for all $x\in L$. Then $f<g$, but obviously there exists no $h\in\Sup(L,L)_{\leq 1_L}$ with $f=g\cdot h$.

    In particular, when $L=C_3=\{\bot,e,\top\}$ is the three-chain, it is not difficult to see that $\Sup(C_3,C_3)_{\leq 1_{C_3}}$ is the following complete lattice:
    $$\bfig
    \morphism/-/<0,300>[\circ`\circ;]
    \morphism(0,300)/-/<-300,200>[\circ`\circ;]
    \morphism(0,300)/-/<300,200>[\circ`\circ;]
    \morphism(-300,500)/-/<300,200>[\circ`\circ;]
    \morphism(300,500)/-/<-300,200>[\circ`\circ;]
    \place(0,-80)[\bot=(\bot,\bot,\bot)]
    \place(-300,250)[a=(\bot,\bot,e)]
    \place(-630,500)[b=(\bot,\bot,\top)]
    \place(630,500)[c=(\bot,e,e)]
    \place(0,780)[\top=(\bot,e,\top)]
    \efig$$
    In fact, the subquantale on the four-chain $C_4=\{\bot,a,b,\top\}\subseteq\Sup(C_3,C_3)_{\leq 1_{C_3}}$ is also non-commutative, non-divisible and integral, which is the simplest complete lattice that can be endowed with such quantale structures. Explicitly, multiplications in the quantale $(C_4,\cdot,\top)$ are given by
    $$a\cdot a=b\cdot a=\bot,\quad a\cdot b=a\quad\text{and}\quad b\cdot b=b,$$
    and $(C_4,\cdot,\top)$ is non-divisible since $a$ is not right-divisible by $b$, although $a<b$ holds.

\end{enumerate}
\end{exmps}

\begin{defn} \label{Q-relation}
A \emph{$\FQ$-relation} (or, \emph{fuzzy relation}) $\phi\colon X\rto Y$ between $\FQ$-subsets is a map $\phi\colon X\times Y\to\FQ$ such that $\phi(x,y)\in\DQ(|x|,|y|)$, i.e.,
\begin{equation} \label{QFRel-def}
(\phi(x,y)\ldd|x|)\with|x|=\phi(x,y)=|y|\with(|y|\rdd\phi(x,y)),
\end{equation}
for all $x\in X$ and $y\in Y$.
\end{defn}

Note that when $\FQ=\boldsymbol{2}$, the two-element Boolean algebra, $\phi\colon X\rto Y$ in Definition~\ref{Q-relation} reduces to a ``binary relation between (crisp) subsets $A\subseteq X$ and $B\subseteq Y$''. In other words, $\phi\colon X\rto Y$ is actually a \emph{partially defined relation} from $X$ to $Y$.

As for a general $\FQ$, Lemma~\ref{DQ-elements}\,\ref{DQ-elements:bot-q} forces $\phi(x,y)=\bot$ in Equation~\eqref{QFRel-def} whenever $|x|=\bot$ or $|y|=\bot$. Hence, with the value $\phi(x,y)$ of a $\FQ$-relation $\phi\colon X\rto Y$ interpreted as the degree of $x$ and $y$ being related, Equation \eqref{QFRel-def} can be understood as a many-valued reformulation of ``$x$ and $y$ are related only if $x$ is in the $\FQ$-subset $(X,|\text{-}|)$ of $X$ and $y$ is in the $\FQ$-subset $(Y,|\text{-}|)$ of $Y$''.

Recall that a $\FQ$-relation between (crisp) sets is nothing but a map $\phi\colon X\times Y\to\FQ$. Since every (crisp) set $X$ can be regarded as a $\FQ$-subset in which $|x|=e$ for all $x\in X$, the following diagram illustrates the chain of generalization from ``binary relations between (crisp) sets'' to ``$\FQ$-relations between $\FQ$-subsets'':
\tikzset{
box/.style={rectangle,
minimum width=30pt, minimum height=15pt, inner sep=6pt,draw}
}
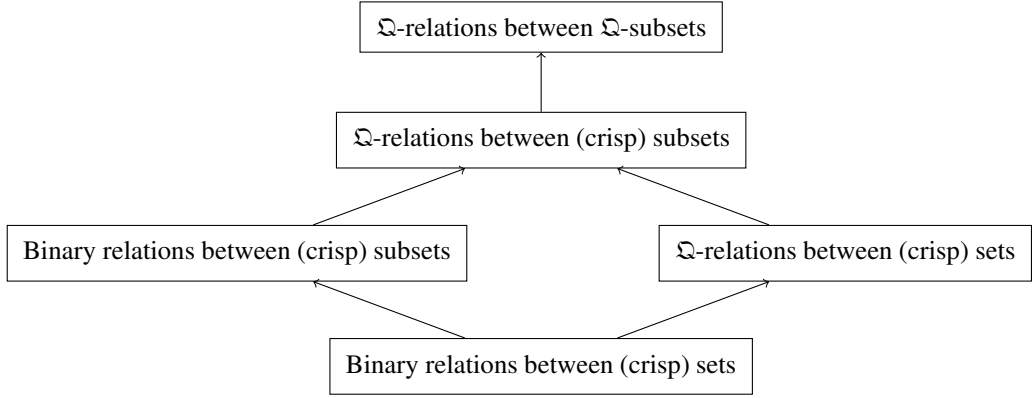
\begin{equation} \label{fuzzy-relation-generalization}
\begin{tikzpicture}
\node[box] (R) at (0,0) {Binary relations between (crisp) sets};
\node[box] (RS) at (-4,1.5) {Binary relations between (crisp) subsets};
\node[box] (Q) at (4,1.5) {$\FQ$-relations between (crisp) sets};
\node[box] (QS) at (0,3) {$\FQ$-relations between (crisp) subsets};
\node[box] (RQS) at (0,4.5) {$\FQ$-relations between $\FQ$-subsets};
\draw[->] (R)--(RS);
\draw[->] (R)--(Q);
\draw[->] (RS)--(QS);
\draw[->] (Q)--(QS);
\draw[->] (QS)--(RQS);
\end{tikzpicture}
\end{equation}

Explicitly, for any $\FQ$-relation $\phi\colon X\rto Y$ between $\FQ$-subsets:
\begin{enumerate}[label=(\arabic*)]
\item If $|x|=|y|=e$ and $\phi(x,y)\in\{\bot,e\}$ for all $x\in X$ and $y\in Y$, then $\phi$ can be identified with the binary relation $R_{\phi}$ between (crisp) sets $X$ and $Y$, given by $x\mathrel{R_{\phi}}y\iff\phi(x,y)=e$ for all $x\in X$ and $y\in Y$.
\item If $|x|,|y|,\phi(x,y)\in\{\bot,e\}$ for all $x\in X$ and $y\in Y$, then $\phi$ can be identified with the binary relation $R_{\phi}$ between (crisp) subsets
    $$\supp X=\{x\in X\mid |x|\neq\bot\}\subseteq X\quad\text{and}\quad\supp Y=\{y\in Y\mid |y|\neq\bot\}\subseteq Y,$$
    given by $x\mathrel{R_{\phi}} y\iff \phi(x,y)=e$ for all $x\in\supp X$ and $y\in\supp Y$.
\item If $|x|=|y|=e$ for all $x\in X$ and $y\in Y$, then Lemma~\ref{DQ-elements}\,\ref{DQ-elements:e-e} shows that $\DQ(|x|,|y|)=\FQ$ for all $x\in X$ and $y\in Y$, and hence $\phi$ is just a map $\phi\colon X\times Y\to\FQ$; that is, a $\FQ$-relation between (crisp) sets $X$ and $Y$.
\item If $|x|,|y|\in\{\bot,e\}$ and $\phi(x,y)\in\FQ$ for all $x\in X$ and $y\in Y$, by Lemma~\ref{DQ-elements}\,\ref{DQ-elements:bot-q}\,\ref{DQ-elements:e-e} one sees that $\DQ(|x|,|y|)=\FQ$ if $x\in\supp X$, $y\in\supp Y$ and $\DQ(|x|,|y|)=\{\bot\}$ otherwise; hence, $\phi$ can be identified with a $\FQ$-relation between (crisp) subsets $\supp X\subseteq X$ and $\supp Y\subseteq Y$.
\end{enumerate}

\begin{exmps} \label{Q-relation-exmp} \
\begin{enumerate}[label=(\arabic*)]
\item \label{Q-relation-exmp:id}
    For any $\FQ$-subset $X$, the map $\id_X\colon X\times X\to\FQ$ with
    $$\id_X(x,x')=\begin{cases}
    |x| & \text{if}\ x=x',\\
    \bot & \text{else}
    \end{cases}$$
    defines a $\FQ$-relation $\id_X\colon X\rto X$ (by Lemma~\ref{DQ-elements}\,\ref{DQ-elements:q-q}\,\ref{DQ-elements:bot-p-q}), called the \emph{identity $\FQ$-relation} on $X$.
\item \label{Q-relation-exmp:DQ}
    For any $q\in\FQ$, let $\boldsymbol{1}_q$ denote the singleton $\FQ$-subset $\{*\}$ with $|*|=q$. Then for any $p,q\in\FQ$, each $u\in\DQ(p,q)$ can be regarded as a $\FQ$-relation $u\colon\boldsymbol{1}_p\rto\boldsymbol{1}_q$ with $u(*,*)=u$. In other words, there are as many $\FQ$-relations $\boldsymbol{1}_p\rto\boldsymbol{1}_q$ as elements in $\DQ(p,q)$.
\item \label{Q-relation-exmp:restriction}
    For any $\FQ$-relation $\phi\colon X\rto Y$, $x\in X$ and $y\in Y$, the maps $\phi(x,-)\colon\boldsymbol{1}_{|x|}\times Y\to\FQ$ and ${\phi(-,y)\colon X\times \boldsymbol{1}_{|y|}\to\FQ}$ given by
    $$\phi(x,-)(*,y)=\phi(x,y)\quad\text{and}\quad\phi(-,y)(x,*)=\phi(x,y)$$
    define $\FQ$-relations
    $$\phi(x,-)\colon\boldsymbol{1}_{|x|}\rto Y\quad\text{and}\quad\phi(-,y)\colon X\rto\boldsymbol{1}_{|y|}.$$
    In particular, $\phi(x,y)$ can be regarded as a $\FQ$-relation $\phi(x,y)\colon \boldsymbol{1}_{|x|}\rto\boldsymbol{1}_{|y|}$, which is a special case of \ref{Q-relation-exmp:DQ}.
\end{enumerate}
\end{exmps}

For $\FQ$-relations $\phi\colon X\rto Y$ and $\psi\colon Y\rto Z$, it is straightforward to check that the map $\psi\comp\phi\colon X\times Z\to\FQ$ with
\begin{align}
(\psi\comp\phi)(x,z)&=\bv_{y\in Y}(\psi(y,z)\ldd|y|)\with|y|\with(|y|\rdd\phi(x,y))\nonumber\\
&=\bv_{y\in Y}(\psi(y,z)\ldd|y|)\with\phi(x,y)\nonumber\\
&=\bv_{y\in Y}\psi(y,z)\with(|y|\rdd\phi(x,y)) \label{psi-circ-phi}
\end{align}
defines a $\FQ$-relation
$$\psi\comp\phi\colon X\rto Z,$$
called the \emph{composition} of $\psi$ and $\phi$. \eqref{psi-circ-phi} can be interpreted as a many-valued version of ``$x$ and $z$ are related if, and only if, there exists $y$ in the $\FQ$-subset $(Y,|\text{-}|)$ of $Y$ such that $x$ and $y$ are related, $y$ and $z$ are related''.

\begin{prop} \label{QFRel-comp}
Let $\phi,\phi',\phi_i\colon X\rto Y$, $\psi,\psi_i\colon Y\rto Z$, $\xi\colon Z\rto W$ $(i\in I)$ be $\FQ$-relations between $\FQ$-subsets.
\begin{enumerate}[label={\rm(\arabic*)}]
\item $\xi\comp(\psi\comp\phi)=(\xi\comp\psi)\comp\phi$.
\item $\id_Y\comp\phi=\phi=\phi\comp\id_X$.
\item \label{QFRel-comp:join}
    With the pointwise order inherited from $\FQ$, i.e.,
    $$\phi\leq\phi'\iff\forall x\in X,\ y\in Y\colon\phi(x,y)\leq\phi'(x,y),$$
    $\FQ$-relations from $X$ to $Y$ form a complete lattice $\QFRel(X,Y)$. Moreover, it holds that
    $$\psi\comp\Big(\bv\limits_{i\in I}\phi_i\Big)=\bv\limits_{i\in I}(\psi\comp\phi_i)\quad\text{and}\quad\Big(\bv\limits_{i\in I}\psi_i\Big)\comp\phi=\bv\limits_{i\in I}(\psi_i\comp\phi).$$
\end{enumerate}
\end{prop}

Proposition~\ref{QFRel-comp}\,\ref{QFRel-comp:join} induces Galois connections
$$\bfig
\morphism/@{->}@<4pt>/<1000,0>[\QFRel(Y,Z)`\QFRel(X,Z);-\comp\phi]
\morphism(1000,0)|b|/@{->}@<4pt>/<-1000,0>[\QFRel(X,Z)`\QFRel(Y,Z);-\lda\phi]
\place(500,5)[\mbox{\footnotesize{$\bot$}}]
\efig
\quad\text{and}\quad
\bfig
\morphism/@{->}@<4pt>/<1000,0>[\QFRel(X,Y)`\QFRel(X,Z);\psi\comp -]
\morphism(1000,0)|b|/@{->}@<4pt>/<-1000,0>[\QFRel(X,Z)`\QFRel(X,Y);\psi\rda -]
\place(500,5)[\mbox{\footnotesize{$\bot$}}]
\efig$$
for all $\FQ$-relations $\phi\colon X\rto Y$ and $\psi\colon Y\rto Z$, where the operations $\lda$, $\rda$ are called \emph{left} and \emph{right implications} of $\FQ$-relations, respectively. Explicitly, for any $\xi\colon X\rto Z$, the implications $\xi\lda\phi$ and $\psi\rda\xi$ are given by
\begin{equation} \label{QFRel-imp}
\xi\lda\phi=\bv\{\psi'\colon Y\rto Z\mid\psi'\comp\phi\leq\xi\}\quad\text{and}\quad\psi\rda\xi=\bv\{\phi'\colon X\rto Y\mid\psi\comp\phi'\leq\xi\}.
\end{equation}

\begin{rem} \label{implications-Q-QFRel}
Given $p,q,r\in\FQ$, $u\in\DQ(p,q)$, $v\in\DQ(q,r)$ and $w\in\DQ(p,r)$, since $u$, $v$ and $w$ are themselves elements in $\FQ$, one could compute the implications
$$w\ldd u\quad\text{and}\quad v\rdd w$$
in $\FQ$. On the other hand, if we consider $u\colon\boldsymbol{1}_p\rto\boldsymbol{1}_q$, $v\colon\boldsymbol{1}_q\rto\boldsymbol{1}_r$, $w\colon\boldsymbol{1}_p\rto\boldsymbol{1}_r$ as $\FQ$-relations (see Example~\ref{Q-relation-exmp}\,\ref{Q-relation-exmp:DQ}),
then it is also possible to calculate the implications
$$w\lda u\colon\boldsymbol{1}_q\rto\boldsymbol{1}_r\quad\text{and}\quad v\rda w\colon\boldsymbol{1}_p\rto\boldsymbol{1}_q$$
of $\FQ$-relations. It is not difficult to see that
$$w\lda u=\bv\{v'\in\DQ(q,r)\mid v'\leq w\ldd(q\rdd u)\}\quad\text{and}\quad v\rda w=\bv\{u'\in\DQ(p,q)\mid u'\leq (v\ldd q)\rdd w\};$$
hence, in general $w\lda u\neq w\ldd u$ and $v\rda w\neq v\rdd w$. That is why we distinguish implications in $\FQ$ and those of $\FQ$-relations with different symbols.
\end{rem}

It is straightforward to verify the following calculus of $\FQ$-relations:

\begin{prop} \label{QFRel-comp-imp}
For $\FQ$-relations $\phi\colon X\rto Y$, $\psi\colon Y\rto Z$ and $\xi\colon X\rto Z$, it holds that
$$\psi\comp\phi=\bv_{y\in Y}\psi(y,-)\comp\phi(-,y),\quad\xi\lda\phi=\bw_{x\in X}\xi(x,-)\lda\phi(x,-)\quad\text{and}\quad\psi\rda\xi=\bw_{z\in Z}\psi(-,z)\rda\xi(-,z).$$
\end{prop}

\begin{prop} \label{QFRel-cal}
The following formulas hold for all $\FQ$-relations $\phi$, $\phi_i$, $\psi$, $\psi_i$, $\xi,\xi_i$ $(i\in I)$ between $\FQ$-subsets whenever the compositions and implications make sense{\textup:}
\begin{enumerate}[label={\rm(\arabic*)}]
\item $\psi\comp\phi\leq\xi\iff\psi\leq\xi\lda\phi\iff\phi\leq\psi\rda\xi$.
\item $(\bw\limits_{i\in I}\xi_i)\lda\phi=\bw\limits_{i\in I}(\xi_i\lda\phi)$ and $\psi\rda(\bw\limits_{i\in I}\xi_i)=\bw\limits_{i\in I}(\psi\rda\xi_i)$.
\item $h\lda(\bv\limits_{i\in I}\phi_i)=\bw\limits_{i\in I}(\xi\lda\phi_i)$ and $(\bv\limits_{i\in I}\psi_i)\rda\xi=\bw\limits_{i\in I}(\psi_i\rda\xi)$.
\item $(\xi\lda\psi)\comp(\psi\lda\phi)\leq\xi\lda\phi$ and $(\phi\rda\psi)\comp(\psi\rda\xi)\leq\phi\rda\xi$.
\item $(\xi\lda\phi)\lda\psi=h\lda(\psi\comp\phi)$ and $\phi\rda(\psi\rda\xi)=(\psi\comp\phi)\rda\xi$.
\item $(\psi\rda\xi)\lda\phi=\psi\rda(\xi\lda\phi)$.
\item $(\xi\lda\phi)\comp\phi\leq\xi$ and $\psi\comp(\psi\rda\xi)\leq\xi$.
\item $\xi\comp(\psi\lda\phi)\leq(\xi\comp\psi)\lda\phi$ and $(\psi\rda\xi)\comp\phi\leq\psi\rda(\xi\comp\phi)$.
\end{enumerate}
\end{prop}

Recall that an \emph{ordered category} \cite{Hofmann2014}, as a special kind of a \emph{$2$-category} \cite{MacLane1998}, is a category $\CC$ whose hom-sets $\CC(X,Y)$ are equipped with a preorder ``$\leq$'', such that
$$v\leq v'\implies w\comp v\comp u\leq w\comp v'\comp u$$
holds for all morphisms $u\colon X\to Y$, $v,v'\colon Y\to Z$ and $w\colon Z\to W$ in $\CC$.

Proposition~\ref{QFRel-comp} in fact shows that $\FQ$-subsets and $\FQ$-relations constitute an ordered category $\QFRel$ which, moreover, is a \emph{quantaloid} \cite{Rosenthal1996}. Explicitly, a quantaloid is a category $\CC$ in which every hom-set $\CC(X,Y)$ is a complete lattice, with
$$v\comp\Big(\bv_{i\in I}u_i\Big)=\bv_{i\in I}(v\comp u_i)\quad\text{and}\quad\Big(\bv_{i\in I}v_i\Big)\comp u=\bv_{i\in I}(v_i\comp u)$$
holding for all morphisms $u,u_i\colon X\to Y$ and $v,v_i\colon Y\to Z$ $(i\in I)$ in $\CC$. The properties of $\FQ$-relations presented in Proposition~\ref{QFRel-cal} are valid for morphisms in any quantaloid $\CC$.

%

\section{Preordered fuzzy sets valued in a quantale} \label{Preordered-fuzzy-sets}

\subsection{$\FQ$-preordered $\FQ$-subsets}

Let $\al\colon X\rto X$ be a $\FQ$-relation on a $\FQ$-subset $X$. Then
\begin{itemize}
\item $\al$ is \emph{reflexive} if $\id_X\leq\al$;
\item $\al$ is \emph{transitive} if $\al\comp\al\leq\al$.
\end{itemize}

\begin{defn} \label{Q-preorder}
A \emph{$\FQ$-preorder} on a $\FQ$-subset $X$ is a reflexive and transitive $\FQ$-relation $\al\colon X\rto X$. The pair $(X,\al)$ is called a \emph{$\FQ$-preordered $\FQ$-subset}.
\end{defn}

In elementary words, a map $\al\colon X\times X\to\FQ$ defines a $\FQ$-preorder on a $\FQ$-subset $X$ if
\begin{enumerate}[label=($\FQ$P\arabic*),leftmargin=3.8em]
\item \label{QPd}
    $(\al(x,y)\ldd |x|)\with |x|=\al(x,y)=|y|\with(|y|\rdd\al(x,y))$,
\item \label{QPr}
    $|x|\leq\al(x,x)$,
\item \label{QPt}
    $(\al(y,z)\ldd |y|)\with\al(x,y)=\al(y,z)\with(|y|\rdd\al(x,y))\leq\al(x,z)$
\end{enumerate}
for all $x,y,z\in X$. These conditions can be intuitively interpreted as:
\begin{enumerate}[label=(\arabic*)]
\item $x\leq y$ only if $x$ and $y$ are both in the $\FQ$-subset $(X,|\text{-}|)$;
\item if $x$ is in the $\FQ$-subset $(X,|\text{-}|)$, then $x\leq x$;
\item if there exists $y$ in the $\FQ$-subset $(X,|\text{-}|)$ such that $y\leq z$ and $x\leq y$, then $x\leq z$.
\end{enumerate}

\begin{rem} \label{QFOrd-CT}
A unital quantale $\FQ$ gives rise to a quantaloid $\DQ$ \cite{Hoehle2011a,Pu2012,Stubbe2014} with the following data:
\begin{enumerate}[label=(\arabic*)]
\item objects in $\DQ$ are elements of $\FQ$;
\item a morphism $u\colon p\to q$ in $\DQ$ is an element in $\FQ$ right-divisible by $p$ and left-divisible by $q$, i.e., $u\in\DQ(p,q)$;
\item the composition of $u\colon p\to q$ and $v\colon q\to r$ in $\DQ$ is given by
    $$v\comp u=(v\ldd q)\with q\with (q\rdd u)=(v\ldd q)\with u=v\with(q\rdd u);$$
\item the identity morphism on $q$ in $\DQ$ is $q\colon q\to q$.
\end{enumerate}
 The structure of the quantaloid $\DQ$ is extremely clear when $\FQ$ is divisible, in which case each hom-set $\DQ(p,q)$ is exactly the principal lower set generated by $p\wedge q$ (see Lemma~\ref{DQ-elements}\,\ref{DQ-elements:divisible}). From the viewpoint of enriched category theory, a $\FQ$-preordered $\FQ$-subset is precisely a category enriched in the quantaloid $\DQ$; we refer to \cite{Heymans2010,Rosenthal1996,Shen2016b,Stubbe2005,Stubbe2014} for the theory of quantaloid-enriched categories.
\end{rem}

Note that the same map $\al\colon X\times X\to\FQ$ can define $\FQ$-preorders on different $\FQ$-subsets over the same (crisp) set $X$. In particular, we have the following:

\begin{prop} \label{Q-Preorder-different-Q-subset}
Let $\al$ be a $\FQ$-preorder on a $\FQ$-subset $(X,|\text{-}|)$.
\begin{enumerate}[label={\rm(\arabic*)}]
\item \label{Q-Preorder-different-Q-subset:interpolate}
    If $(X,|\text{-}|')$ is another $\FQ$-subset with
    \begin{equation} \label{x-prime-alxx}
    |x|\leq|x|'\leq\al(x,x)
    \end{equation}
    for all $x\in X$, then $\al$ is also a $\FQ$-preorder on $(X,|\text{-}|')$.
\item \label{Q-Preorder-different-Q-subset:integral}
    If $\FQ$ is integral, then $\al(x,x)=|x|$ for each $x\in X$.
\end{enumerate}
\end{prop}

\begin{proof}
\ref{Q-Preorder-different-Q-subset:interpolate} $(X,|\text{-}|',\al)$ obviously satisfies \ref{QPr}. For \ref{QPd}, let $x,y\in X$. By applying \eqref{x-prime-alxx} and \ref{QPd}, \ref{QPt} for $(X,|\text{-}|,\al)$ one has
$$\al(x,y)=(\al(x,y)\ldd|x|)\with|x|\leq (\al(x,y)\ldd|x|)\with|x|'\leq(\al(x,y)\ldd |x|)\with\al(x,x)=\al(x,y),$$
which proves the right-divisibility of $\al(x,y)$ by $|x|'$, and its left-divisibility by $|y|'$ can be checked similarly. As for \ref{QPt}, just note that
$$\al(y,z)\with(|y|'\rdd\al(x,y))\leq\al(y,z)\with(|y|\rdd\al(x,y))\leq \al(x,z)$$
for all $x,y,z\in X$, by \eqref{x-prime-alxx} and \ref{QPt} for $(X,|\text{-}|,\al)$.

\ref{Q-Preorder-different-Q-subset:integral} If $\FQ$ is integral, then it follows from \ref{QPd} that $\al(x,x)\leq|x|$ for each $x\in X$, and thus $\al(x,x)=|x|$ by \ref{QPr}.
\end{proof}


\begin{rem} \label{Q-preorder-divisible}
As we remarked in \ref{QFOrd-CT}, in the case that $\FQ$ is a divisible quantale, Lemma~\ref{DQ-elements}\,\ref{DQ-elements:divisible} simplifies the condition \ref{QPd} for a $\FQ$-preorder $\al$ on a $\FQ$-subset $X$ to
$$\al(x,y)\leq|x|\wedge|y|$$
for all $x,y\in X$. Since divisible quantales are necessarily integral, with Proposition~\ref{Q-Preorder-different-Q-subset}\,\ref{Q-Preorder-different-Q-subset:integral} one deduces that a {\protect\linebreak}$\FQ$-preordered $\FQ$-subset is exactly a pair $(X,\al)$, where $X$ is a (crisp) set and $\al\colon X\times X\to\FQ$ is a map, such that
\begin{enumerate}[label=(DP\arabic*),leftmargin=3.2em]
\item \label{DPr}
    $\al(x,y)\leq\al(x,x)\wedge\al(y,y)$,
\item \label{DPt}
    $(\al(y,z)\ldd\al(y,y))\with\al(x,y)=\al(y,z)\with(\al(y,y)\rdd\al(x,y))\leq\al(x,z)$
\end{enumerate}
for all $x,y,z\in X$. When $\FQ$ is commutative and divisible, the conditions \ref{DPr} and \ref{DPt} were first presented by H{\"o}hle-Kubiak to formulate their \emph{pre-$\FQ$-sets} (see \cite[Example 3.4]{Hoehle2011a}) and by Pu-Zhang in the definition of their {\protect\linebreak}\emph{$\FQ$-valued preordered sets} (see \cite[Definition 3.4]{Pu2012}), which are both precisely $\FQ$-preordered $\FQ$-subsets in our sense.\footnote{Pre-$\FQ$-sets were defined in \cite[Definition 3.1]{Hoehle2011a} for a general quantale $\FQ$, but they may not be identified with our $\FQ$-preordered $\FQ$-subsets, unless $\FQ$ is commutative and divisible.}

However, if $\FQ$ is non-divisible, one may find pairs $(X,\al)$ satisfying \ref{DPr} and \ref{DPt} but fail to be $\FQ$-preordered $\FQ$-subsets. For example, let $\FQ=(C_4,\cdot,\top)$ be the non-divisible integral quantale introduced in Example~\ref{Q-exmp}\,\ref{Q-exmp:nc-nd-i}, and let $X=\{x,y\}$ with $\al\colon X\times X\to\FQ$ given by
$$\renewcommand\arraystretch{1.3}
\setlength\doublerulesep{0pt}
\begin{tabular}{c||c|c}
$\al(\text{-},\text{-})$ & $x$ & $y$\\
\hline\hline
$x$ & $b$ &$a$ \\
\hline
$y$ & $\bot$ & $b$
\end{tabular}$$
Then $(X,\al)$ satisfies \ref{DPr} and \ref{DPt}, but in order for $\al$ to become a $\FQ$-preorder on the $\FQ$-subset $(X,|\text{-}|)$ with $|x|=|y|=\al(x,x)=\al(y,y)=b$ (see Proposition~\ref{Q-Preorder-different-Q-subset}\,\ref{Q-Preorder-different-Q-subset:integral}), $\al(x,y)=a$ should belong to $\DQ(|x|,|y|)=\DQ(b,b)$, which cannot be true since $a$ is not right-divisible by $b$.
\end{rem}

Each $\FQ$-preordered $\FQ$-subset $(X,\al)$ admits a natural underlying preorder on $X$ given by
$$x\leq y\iff |x|=|y|\quad\text{and}\quad |x|\leq\al(x,y).$$
$(X,\al)$ is said to be \emph{separated} if $(X,\leq)$ is a partial order; that is, $x=y\iff x\leq y\ \text{and}\ y\leq x$.

\begin{rem} \label{underlying-preorder-bottom}
Given a $\FQ$-preordered $\FQ$-subset $(X,\al)$, it is easy to see that $x\leq y$ in the underlying preorder whenever $|x|=|y|=\bot$; that is, the underlying preorder of $(X,\al)$ always endows the set
$$X_{\bot}=\{x\in X\mid |x|=\bot\}$$
with the \emph{indiscrete} preorder. Consequently, if $(X,\al)$ is separated, then there is at most one element $x\in X$ with $|x|=\bot$.
\end{rem}

\begin{defn}
A membership-preserving map $f\colon (X,\al)\to(Y,\be)$ between $\FQ$-preordered $\FQ$-subsets is \emph{$\FQ$-order-preserving} if
$$\al(x,x')\leq\be(fx,fx')$$
for all $x,x'\in X$. A $\FQ$-order-preserving map $f\colon (X,\al)\to(Y,\be)$ is \emph{fully faithful} if
$$\al(x,x')=\be(fx,fx')$$
for all $x,x'\in X$.
\end{defn}

With the pointwise (pre)order of $\FQ$-order-preserving maps $f,g\colon(X,\al)\to(Y,\be)$ given by
$$f\leq g\iff\forall x\in X\colon fx\leq gx\iff\forall x\in X\colon |x|\leq\be(fx,gx),$$
$\FQ$-preordered $\FQ$-subsets and $\FQ$-order-preserving maps constitute an ordered category $\QFOrd$. Fully faithful and bijective $\FQ$-order-preserving maps are clearly isomorphisms in $\QFOrd$.

\begin{exmps} \label{QFOrd-exmp}  \
\begin{enumerate}[label=(\arabic*)]
\item \label{QFOrd-exmp:discrete}
    Each $\FQ$-subset $X$ is equipped with a \emph{discrete} $\FQ$-preorder $\id_X\colon X\rto X$. In particular, the singleton $\FQ$-subset $\boldsymbol{1}_q$ (see Example~\ref{Q-relation-exmp}\,\ref{Q-relation-exmp:DQ}) is always assumed to be equipped with the discrete $\FQ$-preorder.
\item \label{QFOrd-exmp:POrd}
    If $\FQ=\boldsymbol{2}$, then Definition~\ref{Q-preorder} gives \emph{partially defined preordered sets} (see \cite{Shen2015}); that is, (crisp) sets $X$ equipped with a preorder on a (crisp) subset $\supp X\subseteq X$. A morphism $f\colon X\to Y$ in
    $$\POrd:=\boldsymbol{2}\text{-}\FOrd$$
    is a map $f\colon X\to Y$ satisfying $\supp X=f^{-1}(\supp Y)$, whose restriction $f|_{\supp X}\colon\supp X\to\supp Y$ is order-preserving. It should be pointed out that $X\in\POrd$ is \emph{not} a preordered set as long as $\supp X\subsetneq X$ due to the failure of reflexivity for elements in $X\setminus\supp X$, but the underlying preorder of $X$ is indeed a preorder on $X$ by assigning $\supp X$ with its original order and $X\setminus\supp X$ with the indiscrete preorder (see Remark~\ref{underlying-preorder-bottom}).
\item \label{QFOrd-exmp:skew}
    Every frame $\Om=(\Om,\wedge,\top)$ is a commutative and divisible quantale. Moreover,
    $$(q\ra v)\wedge u=v\wedge(q\ra u)=v\wedge u$$
    whenever $u,v\leq q$ in $\Om$. Then by Remark~\ref{Q-preorder-divisible}, an $\Om$-preordered $\Om$-subset becomes a (crisp) set $X$ equipped with a map $\al\colon X\times X\to\Om$, such that
    $$\al(x,y)\leq\al(x,x)\wedge\al(y,y)\quad\text{and}\quad\al(y,z)\wedge\al(x,y)\leq\al(x,z)$$
    for all $x,y,z\in X$. Therefore, $\Om$-preordered $\Om$-subsets are precisely \emph{skew $\Om$-sets} in the sense of Borceux-Cruciani \cite{Borceux1998}. In particular, skew $\Om$-sets $(X,\al)$ satisfying $\al(x,y)=\al(y,x)$ for all $x,y\in X$ are exactly \emph{$\Om$-sets} originally defined by Fourman-Scott \cite{Fourman1979}.
\item \label{QFOrd-exmp:ParMet}
    Since Lawvere's quantale $\FQ=([0,\infty]^{\op},+,0)$ is divisible, it follows from Remark~\ref{Q-preorder-divisible} that a $\FQ$-preordered $\FQ$-subset is exactly a pair $(X,\al)$, where $X$ is a (crisp) set and $\al\colon X\times X\to[0,\infty]$ is a map, such that
    $$\al(x,x)\vee\al(y,y)\leq\al(x,y)\quad\text{and}\quad\al(x,z)\leq\al(y,z)-\al(y,y)+\al(x,y)$$
    for all $x,y,z\in X$; that is to say, $(X,\al)$ is a (\emph{generalized}) \emph{partial metric space} (see \cite{Bukatin2009,Hoehle2011a,Kuenzi2006a,Matthews1994,Pu2012}). Morphisms $f\colon(X,\al)\to(Y,\be)$ between partial metric spaces are non-expanding maps; that is, maps $f\colon X\to Y$ satisfying
    $$\al(x,x)=\be(fx,fx)\quad\text{and}\quad\al(x,x')\geq\be(fx,fx')$$
    for all $x,x'\in X$.
\item \label{QFOrd-exmp:Q}
    Considering $\FQ$ itself as a $\FQ$-subset with $|q|=e$ for all $q\in\FQ$, there is an intrinsic $\FQ$-preorder $\al\colon\FQ\rto\FQ$ on $(\FQ,|\text{-}|)$ with
    $$\al(p,q)=q\ldd p$$
    for all $p,q\in\FQ$, whose underlying preorder coincides with the given order on $\FQ$. From Proposition~\ref{Q-Preorder-different-Q-subset}\,\ref{Q-Preorder-different-Q-subset:interpolate} we know that for any map $|\text{-}|'\colon\FQ\to\FQ$ with $e\leq|q|'\leq q\ldd q$ for all $q\in\FQ$, $\al$ is a $\FQ$-preorder on $(\FQ,|\text{-}|')$. In particular, $\al$ is a $\FQ$-preorder on $(\FQ,|\text{-}|_{\al})$ with $|q|_{\al}=q\ldd q$ for all $q\in\FQ$, and the underlying preorder of $(\FQ,|\text{-}|_{\al},\al)$ is, in general, coarser than the given order on $\FQ$.

    As the $\FQ$-preordered $\FQ$-subsets $(\FQ,|\text{-}|,\al)$ and $(\FQ,|\text{-}|_{\al},\al)$ coincide if, and only if, the quantale $\FQ$ is integral, the simplest example in which they differ is the quantale $\FQ=(C_3,\with,e)$ introduced in Example~\ref{Q-exmp}\,\ref{Q-exmp:c-ni}. Indeed, there are precisely four $\FQ$-subsets of $C_3$ on which $\al$ is a $\FQ$-preorder:
    \begin{itemize}
    \item $|\bot|_1=|e|_1=|\top|_1=e$, whose underlying preorder is the given order on $C_3$;
    \item $|\bot|_2=|e|_2=e$ and $|\top|_2=\top$, whose underlying preorder on $C_3$ is given by $\bot\leq e$;
    \item $|\bot|_3=\top$ and $|e|_3=|\top|_3=e$, whose underlying preorder on $C_3$ is given by $e\leq\top$;
    \item $|\bot|_4=|\top|_4=\top$ and $|e|_4=e$, whose underlying preorder on $C_3$ is given by $\bot\leq\top$.
    \end{itemize}
\end{enumerate}
\end{exmps}

\subsection{Comparison: $\FQ$-preordered {\textup(}crisp{\textup)} sets}

As a counterpart of Diagram~\eqref{fuzzy-relation-generalization} for $\FQ$-relations between $\FQ$-subsets, the following diagram explains how one generalizes step by step from preordered sets to $\FQ$-preordered $\FQ$-subsets:
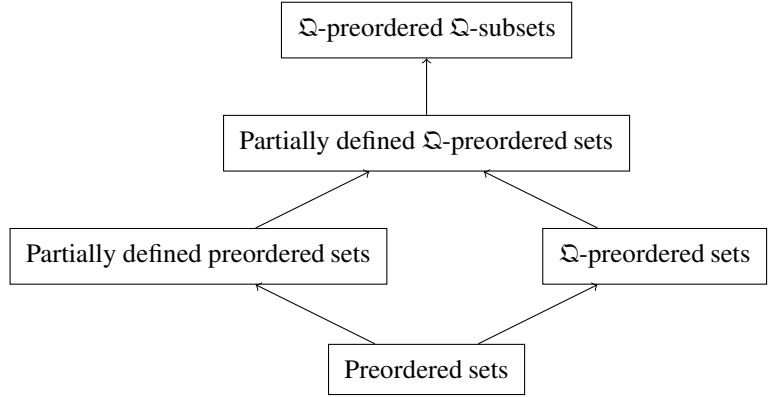
\begin{equation} \label{fuzzy-order-generalization}
\begin{tikzpicture}
\node[box] (Ord) at (0,0) {Preordered sets};
\node[box] (POrd) at (-3,1.5) {Partially defined preordered sets};
\node[box] (QOrd) at (3,1.5) {$\FQ$-preordered sets};
\node[box] (QPOrd) at (0,3) {Partially defined $\FQ$-preordered sets};
\node[box] (QFOrd) at (0,4.5) {$\FQ$-preordered $\FQ$-subsets};
\draw[->] (Ord)--(POrd);
\draw[->] (Ord)--(QOrd);
\draw[->] (POrd)--(QPOrd);
\draw[->] (QOrd)--(QPOrd);
\draw[->] (QPOrd)--(QFOrd);
\end{tikzpicture}
\end{equation}

Recall that a $\FQ$-preorder on a (crisp) set $X$ is given by a map $\al\colon X\times X\to\FQ$ satisfying
\begin{enumerate}[label=(\arabic*)]
\item $e\leq\al(x,x)$,
\item $\al(y,z)\with\al(x,y)\leq\al(x,z)$
\end{enumerate}
for all $x,y,z\in X$. Since a (crisp) set $X$ can be considered as a $\FQ$-subset with $|x|=e$ for all $x\in X$, for any $\FQ$-preordered $\FQ$-subset $(X,\al)$:
\begin{enumerate}[label=(\arabic*)]
\item If $|x|=e$ and $\al(x,y)\in\{\bot,e\}$ for all $x,y\in X$, then $\al$ can be identified with the underlying preorder on $X$ and, hence, $(X,\al)$ is identified with a \emph{preordered set}.
\item If $|x|,\al(x,y)\in\{\bot,e\}$ for all $x,y\in X$, then $\al$ can be identified with the underlying preorder defined on the (crisp) subset $\supp X\subseteq X$ and, hence, $(X,\al)$ is identified with a \emph{partially defined preordered set}.
\item If $|x|=e$ for all $x\in X$, then $\DQ(|x|,|y|)=\FQ$ for all $x,y\in X$; that is, $(X,\al)$ is just a \emph{$\FQ$-preordered set}.
\item If $|x|\in\{\bot,e\}$ and $\al(x,y)\in\FQ$ for all $x,y\in X$, then ${\DQ(|x|,|y|)=\FQ}$ if $x,y\in\supp X$ and $\DQ(|x|,|y|)=\{\bot\}$ otherwise; hence, $\al$ can be identified with a $\FQ$-preorder on the (crisp) subset $\supp X\subseteq X$, which turns $(X,\al)$ into a \emph{partially defined $\FQ$-preordered set}.
\end{enumerate}

Hence, one has the following full embeddings of ordered categories, where $\Ord$, $\POrd$, $\QOrd$, $\QPOrd$ are all full subcategories of $\QFOrd$, consisting of preordered sets, partially defined preordered sets, $\FQ$-preordered sets, partially defined $\FQ$-preordered sets, respectively:
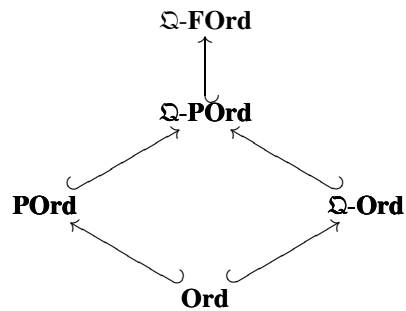
\begin{equation} \label{QOrd-embed}
\bfig
\morphism(600,350)/_(->/<0,350>[\QPOrd`\QFOrd;]
\Atriangle/<-_)`<-^)`/<600,350>[\QPOrd`\POrd`\QOrd;``]
\Vtriangle(0,-350)/`<-^)`<-_)/<600,350>[\POrd`\QOrd`\Ord;``]
\efig
\end{equation}

Indeed, all the embeddings in (\ref{QOrd-embed}) are coreflective. To see this, one first observes an easy fact:

\begin{lem} \label{QSFOrd-QTFOrd-cor}
Let $S\subseteq\FQ$ and $\FQ_S\text{-}\FOrd$ denote the full subcategory of $\QFOrd$ consisting of $\FQ$-preordered $\FQ$-subsets $(X,\al)$ with $|x|\in S$ for all $x\in X$. If $S\subseteq T\subseteq\FQ$, then $\FQ_S\text{-}\FOrd$ is a coreflective subcategory of $\FQ_T\text{-}\FOrd$ with the coreflector sending each $(X,\al)\in\FQ_T\text{-}\FOrd$ to the set
$$X_S=\{x\in X\mid |x|\in S\}$$
equipped with the membership map and the $\FQ$-preorder inherited from $(X,\al)$.
\end{lem}

As a special case of Lemma~\ref{QSFOrd-QTFOrd-cor}, one has \ref{all-emb-coref:Ord-QPOrd} and \ref{all-emb-coref:QPOrd-QFOrd} of the following proposition, while \ref{all-emb-coref:POrd-QPOrd} can be checked easily:

\begin{thm} \label{all-emb-coref}
All the full embeddings of ordered categories in \eqref{QOrd-embed} are coreflective. To be specific{\textup:}
\begin{enumerate}[label={\rm (\arabic*)}]
\item \label{all-emb-coref:POrd-QPOrd}
    $\POrd$ is a coreflective subcategory of $\QPOrd$, with the coreflector sending each $(X,\al)\in\QPOrd$ to its underlying preorder. Similarly, $\Ord$ is a coreflective subcategory of $\QOrd$.
\item \label{all-emb-coref:Ord-QPOrd}
    $\QOrd$ is a coreflective subcategory of $\QPOrd$, with the coreflector sending each $(X,\al)\in\QPOrd$ to the set
    $$X_{\{e\}}=\{x\in X\mid |x|=e\}$$
    equipped with the $\FQ$-preorder inherited from $(X,\al)$. In particular, $\Ord$ is a coreflective subcategory of $\POrd$.
\item \label{all-emb-coref:QPOrd-QFOrd}
    $\QPOrd$ is a coreflective subcategory of $\QFOrd$, with the coreflector sending each $(X,\al)\in\QFOrd$ to the set
    $$X_{\{e,\bot\}}=\{x\in X\mid |x|=e\ \text{or}\ |x|=\bot\}$$
    equipped with the $\FQ$-preorder inherited from $(X,\al)$.
\end{enumerate}
\end{thm}

\begin{rem} \label{Q-valued-preorder}
The different notions of $\FQ$-preorders involved in Diagram~\eqref{fuzzy-order-generalization} must be carefully distinguished from \emph{$\FQ$-valued preordered sets}, as considered by H{\"o}hle (see \cite[Definition 3.2]{Hoehle2015}). Explicitly, given a (crisp) set $X$, a map $\al\colon X\times X\to\FQ$ is called a \emph{$\FQ$-valued preorder} on $X$, if
\begin{enumerate}[label=(\arabic*)]
\item (divisibility) $(\al(x,y)\ldd\al(y,y))\with\al(y,y)=\al(x,y)=\al(x,x)\with(\al(x,x)\rdd\al(x,y))$,
\item (transitivity) $\al(x,y)\with(\al(y,y)\rdd\al(y,z))=(\al(x,y)\ldd\al(y,y))\with\al(y,z)\leq\al(x,z)$
\end{enumerate}
for all $x,y,z\in X$. Indeed, let $\FQ^{\tau}$ denote the \emph{conjugate} of $\FQ$, i.e., the unital quantale whose underlying complete lattice is the same as $\FQ$ and whose multiplication $\with^{\tau}$ satisfies $p\mathrel{\with^{\tau}}q=q\with p$ for all $p,q\in\FQ$, then a $\FQ$-valued preorder $\al$ on $X$ is precisely a $\FQ^{\tau}$-preorder on the $\FQ^{\tau}$-subset $(X,|\text{-}|)$ given by
$$|x|=\al(x,x)$$
for all $x\in X$. Conversely, Proposition~\ref{Q-Preorder-different-Q-subset}\,\ref{Q-Preorder-different-Q-subset:interpolate} indicates that each $\FQ$-preorder $\al$ defined on a $\FQ$-subset $(X,|\text{-}|)$ determines a unique $\FQ^{\tau}$-valued preorder on the crisp set $X$.

In the case that $\FQ$ is integral, Proposition~\ref{Q-Preorder-different-Q-subset}\,\ref{Q-Preorder-different-Q-subset:integral} shows that $\FQ$-preorders defined on a $\FQ$-subset $(X,|\text{-}|)$ coincide with $\FQ^{\tau}$-valued preorders defined on the crisp set $X$. In particular, if $\FQ$ is divisible, $\FQ^{\tau}$-valued preordered sets as characterized by \cite[Proposition 3.3]{Hoehle2015} are exactly $\FQ$-preordered $\FQ$-subsets defined by \ref{DPr} and \ref{DPt} in Remark~\ref{Q-preorder-divisible}.

However, without the hypothesis of integrality on $\FQ$, for a given crisp set $X$ one may construct more $\FQ$-preordered $\FQ$-subsets $(X,|\text{-}|,\al)$ than $\FQ^{\tau}$-valued preorders on $X$, as Example~\ref{QFOrd-exmp}\,\ref{QFOrd-exmp:Q} shows.
\end{rem}

We end this subsection with an interesting comparison of the number of $\FQ$-preordered sets, $\FQ$-valued preordered sets (in the sense of H{\"o}hle) and $\FQ$-preordered $\FQ$-subsets that can be defined on a singleton set $\{*\}$:

\begin{rem} \label{singleton-comparison} \
\begin{enumerate}[label=(\arabic*)]
\item \label{singleton-comparison:QOrd}
    Each idempotent element $q\in\FQ$ greater than or equal to $e$ determines a $\FQ$-preordered set $(\{*\},\al_q)$ with $\al_q(*,*)=q$, and vice versa; that is, there are as many $\FQ$-preorders on $\{*\}$ as idempotent elements in $\FQ$ greater than or equal to $e$, among which there are at least $e$ and $\top$. Hence, $\FQ$ is integral if, and only if, there is precisely one $\FQ$-preorder on $\{*\}$.
\item \label{singleton-comparison:QvOrd}
    Each element $q\in\FQ$ obviously determines a unique $\FQ$-valued preorder on $\{*\}$ in the sense of H{\"o}hle, and vice versa; that is, the number of H{\"o}hle's $\FQ$-valued preorders on $\{*\}$ equals to the cardinality of $\FQ$.
\item \label{singleton-comparison:QFOrd}
    It follows immediately from Example~\ref{QFOrd-exmp}\,\ref{QFOrd-exmp:discrete} that each singleton $\FQ$-subset $\boldsymbol{1}_q$ $(q\in\FQ)$ can be equipped with at least one $\FQ$-preorder, i.e., the discrete one. In particular, since $\DQ(e,e)=\FQ$ by Lemma~\ref{DQ-elements}\,\ref{DQ-elements:e-e}, similar to \ref{singleton-comparison:QOrd} we see that there are as many $\FQ$-preorders on $\boldsymbol{1}_e$ as idempotent elements in $\FQ$ greater than or equal to $e$, among which there are at least $e$ and $\top$. Hence, the combination of Lemma~\ref{DQ-elements}\,\ref{DQ-elements:integral} and \ref{QPr} shows that $\FQ$ is integral if, and only if, there is precisely one $\FQ$-preorder on $\boldsymbol{1}_q$ for every $q\in\FQ$.
\end{enumerate}
\end{rem}

\subsection{Potential lower {\textup(}upper{\textup)} $\FQ$-subsets}

Each $\FQ$-order-preserving map $f\colon (X,\al)\to(Y,\be)$ induces two $\FQ$-relations
$$f_{\nat}\colon X\rto Y,\quad f_{\nat}(x,y)=\be(fx,y)\quad\text{and}\quad f^{\nat}\colon Y\rto X,\quad f^{\nat}(y,x)=\be(y,fx),$$
called respectively the \emph{graph} and \emph{cograph} of $f$. Obviously, for any $(X,\al)\in\QFOrd$, the identity map $1_X$ is $\FQ$-order-preserving, and
$$\al=(1_X)_{\nat}=1_X^{\nat}.$$
Hence, in order to simplify the notation, from now on we abbreviate a $\FQ$-preordered $\FQ$-subset $(X,\al)$ to $X$, and use
$$1_X^{\nat}\colon X\rto X$$
as the standard notation for the $\FQ$-preorder structure on $X$. In summary, whenever we say ``$X$ is a $\FQ$-preordered $\FQ$-subset'' or ``$X\in\QFOrd$'', it means that $X$ is equipped with
\begin{enumerate}[label=(\arabic*)]
\item a membership map $|\text{-}|\colon X\to\FQ$, and
\item a $\FQ$-relation $1_X^{\nat}\colon X\rto X$ as the $\FQ$-preorder on $X$.
\end{enumerate}

\begin{defn} \label{potential-lower-def}
Let $X$ be a $\FQ$-preordered $\FQ$-subset. A $\FQ$-relation $\mu\colon X\rto\boldsymbol{1}_q$ (resp.\ $\lam\colon\boldsymbol{1}_q\rto X$) is called a \emph{potential lower {\textup(}resp.\ upper{\textup)} $\FQ$-subset} of $X$ if
\begin{equation} \label{PX-def}
\mu\comp 1_X^{\nat}\leq\mu\quad(\text{resp.}\ 1_X^{\nat}\comp\lam\leq\lam).
\end{equation}
\end{defn}

Since the reverse inequality of \eqref{PX-def} is trivial, a potential lower (resp. upper) $\FQ$-subset $\mu\colon X\rto\boldsymbol{1}_q$ (resp.\ $\lam\colon\boldsymbol{1}_q\rto X$) necessarily satisfies
\begin{equation} \label{PX-def-eq}
\mu\comp 1_X^{\nat}=\mu\quad(\text{resp.}\ 1_X^{\nat}\comp\lam=\lam).
\end{equation}
In elementary words, a potential lower $\FQ$-subset of $X$ consists of a map $\mu\colon X\to\FQ$ and an element $q\in\FQ$, such that
\begin{enumerate}[label=(\arabic*)]
\item $\mu(x)=(\mu(x)\ldd |x|)\with |x|$,
\item $\mu(x)= q\with(q\rdd\mu(x))$, and
\item $(\mu(y)\ldd |y|)\with 1_X^{\nat}(x,y)=\mu(y)\with(|y|\rdd1_X^{\nat}(x,y))\leq\mu(x)$
\end{enumerate}
for all $x,y\in X$. If we consider $(X,\mu)$ as a $\FQ$-subset, then $q$ can be interpreted as the degree of $(X,\mu)$ being a lower $\FQ$-subset of $X$, and the above conditions can be translated as:
\begin{enumerate}[label=(\arabic*)]
\item $x$ is in $(X,\mu)$ only if $x$ is in $(X,|\text{-}|)$;
\item the degree of $x$ being in $(X,\mu)$ is less than or equal to $q$;
\item if $x\leq y$ and $y$ is in $(X,\mu)$, then $x$ is in $(X,\mu)$.
\end{enumerate}

Potential lower $\FQ$-subsets of $X$ constitute a $\FQ$-subset $\PX$, called the \emph{$\FQ$-powerset} of $X$, with the membership map sending each $\mu\colon X\rto\boldsymbol{1}_q$ to $|\mu|=q$. There is a natural $\FQ$-preorder on $\PX$ given by
$$1_{\PX}^{\nat}(\mu,\mu')=\mu'\lda\mu$$
for all $\mu,\mu'\in \PX$, which is intuitively the inclusion order of potential lower $\FQ$-subsets.

Dually, potential upper $\FQ$-subsets of $X$ constitute a $\FQ$-preordered $\FQ$-subset $\PdX$, called the \emph{dual $\FQ$-powerset} of $X$, with the membership map sending each $\lam\colon \boldsymbol{1}_q\rto X$ to $|\lam|=q$ and. The natural $\FQ$-preorder on $\PdX$ given by
$$1_{\PdX}^{\nat}(\lam,\lam')=\lam'\rda\lam$$
for all $\lam,\lam'\in \PdX$ is intuitively the \emph{reverse} inclusion order (see Remark~\ref{PdX-order} below) of potential upper $\FQ$-subsets.

\begin{rem} \label{PdX-order}
It is important to note that for any $X\in\QFOrd$, it follows from the definition that the underlying preorder on $\PdX$ is the \emph{reverse} local order of $\QFRel$, i.e.,
$$\lam\leq\lam'\ \text{in}\ \PdX\iff\lam'\leq\lam\ \text{in}\ \QFRel.$$
In order to get rid of the confusion about the symbol ``$\leq$'', we make the convention that ``$\leq$'' between $\FQ$-relations always stands for the local order in $\QFRel$ unless otherwise specified.
\end{rem}

\begin{rem}
If $X\in\QFOrd$ is regarded as a category enriched in the quantaloid $\DQ$ (see Remark~\ref{QFOrd-CT}), then $\mu\in\PX$ is precisely a \emph{presheaf} (also \emph{contravariant presheaf} \cite{Stubbe2005,Stubbe2014}) on $X$, while $\lam\in\PdX$ is exactly a \emph{copresheaf} (also \emph{covariant presheaf}) on $X$.
\end{rem}

\begin{exmps} \label{PX-exmp} \
\begin{enumerate}[label=(\arabic*)]
\item \label{PX-exmp:POrd}
Let $X$ be a partially defined preordered set (see Example~\ref{QFOrd-exmp}\,\ref{QFOrd-exmp:POrd}). Then $\PX$ (resp. $\PdX$) consists of pairs $(A,q)$ $(q=0,1)$, where $A$ is a lower (resp. upper) subset of $\supp X$ if $q=1$, and $A=\varnothing$ if $q=0$.
\item \label{PX-exmp:ParMet}Let $X=(X,\al)$ be a partial metric space (see Example~\ref{QFOrd-exmp}\,\ref{QFOrd-exmp:ParMet}). Then $\PX$ consists of pairs $(\mu,q)$, where ${\mu\colon X\lra[0,\infty]}$ is a map and $q\in[0,\infty]$, such that
    $$\al(x,x)\vee q\leq \mu(x)\leq\al(x,y)+\mu(y)-\al(y,y)$$
    for all $x,y\in X$; dually, such a pair $(\mu,q)\in\PdX$ if
    $$\al(x,x)\vee q\leq\mu(x)\leq\al(y,x)+\mu(y)-\al(y,y)$$
    for all $x,y\in X$.
\item \label{PX-exmp:singleton}
    Let $X=\boldsymbol{1}_p\in\QFOrd$ with $p\in \FQ$. Then it follows from Example~\ref{Q-relation-exmp}\,\ref{Q-relation-exmp:DQ} that any potential lower (resp. upper) $\FQ$-subset of $X$ can be regarded as a $u\in\DQ(p,q)$  (resp. $v\in\DQ(q,p)$), i.e.,
    $$\sP\boldsymbol{1}_p=\{u\in\DQ(p,q)\mid q\in\FQ\}\quad\text{and}\quad\sPd\boldsymbol{1}_p=\{v\in\DQ(q,p)\mid q\in\FQ\}.$$
 With Remark~\ref{implications-Q-QFRel} one may exhibit the natural $\FQ$-preorders on $\sP\boldsymbol{1}_p$ and $\sPd\boldsymbol{1}_p$ as
    \begin{align*}
    &1_{\sP\boldsymbol{1}_p}^{\nat}(u,u')=u'\lda u=\bv\{w\in\DQ(q,r)\mid w\leq u'\ldd(q\rdd u)\}\quad\text{and}\\
    &1_{\sPd\boldsymbol{1}_p}^{\nat}(v,v')=v'\rda v=\bv\{w\in\DQ(q,r)\mid w\leq (v'\ldd r)\rdd v\}
    \end{align*}
    for all $q,r\in\FQ$, $u\in\DQ(p,q)$, $u'\in\DQ(p,r)$, $v\in\DQ(q,p)$ and $v'\in\DQ(r,p)$.
\item \label{PX-exmp:powerset} (Fuzzy powerset of a fuzzy set)
    For each $\FQ$-subset $X\in\Set/\FQ$, the \emph{$\FQ$-powerset} of $X$ is defined as the $\FQ$-powerset of the discrete $\FQ$-preordered $\FQ$-subset $(X,\id_X)\in\QFOrd$, whose elements are \emph{potential $\FQ$-subsets} of $X$, i.e., $\FQ$-relations
    $$\mu\colon X\rto\boldsymbol{1}_q$$
    with $|\mu|=q$ interpreted as the degree of $\mu$ being a $\FQ$-subset of $X$. It should be reminded that the $\FQ$-preorder structure on $\PX$ is \emph{not} discrete, although $X$ is equipped with the discrete $\FQ$-preorder.

    We point out that even if $X$ is a crisp set, its $\FQ$-powerset $\PX$ is different from the crisp set $\FQ^X$ of maps ${X\to\FQ}$, which is also referred to as the $\FQ$-powerset (or fuzzy powerset) of $X$ in the literature:
    \begin{itemize}
    \item $\FQ^X$ is a \emph{crisp set} consisting of \emph{$\FQ$-subsets} of $X$;
    \item $\PX$ is a \emph{$\FQ$-subset} consisting of \emph{potential $\FQ$-subsets} of $X$.
    \end{itemize}
\end{enumerate}
\end{exmps}

Given $X\in\QFOrd$, each $x\in X$ gives rise to a \emph{principal potential lower $\FQ$-subset} (cf. Example~\ref{Q-relation-exmp}\,\ref{Q-relation-exmp:restriction})
$$\sy_X x:=1_X^{\nat}(-,x)\colon X\rto\boldsymbol{1}_{|x|}.$$
It is easy to check that the assignment $x\longmapsto\sy_X x$ defines a fully faithful $\FQ$-order-preserving map ${\sy_X\colon X\to\PX}$, called the \emph{Yoneda embedding}. Dually, the fully faithful \emph{co-Yoneda embedding} $\syd_X\colon X\to\PdX$ sends each $x\in X$ to the \emph{principal potential upper $\FQ$-subset}
$$\syd_X x:=1_X^{\nat}(x,-)\colon \boldsymbol{1}_{|x|}\rto X.$$

\begin{lem}[Yoneda] \label{Yoneda}
For any $X\in\QFOrd$, $\mu\in\PX$ and $\lam\in\PdX$, it holds that
$$\mu=(\sy_X)_{\nat}(-,\mu)=1_{\PX}^{\nat}(\sy_X-,\mu)\quad\text{and}\quad\lam=(\syd_X)^{\nat}(\lam,-)=1_{\PdX}^{\nat}(\lam,\syd_X-).$$
\end{lem}

\subsection{Complete $\FQ$-preordered $\FQ$-subsets}

\begin{defn} \label{sup-def}
Let $X\in\QFOrd$. The \emph{supremum} of a potential lower $\FQ$-subset $\mu\colon X\rto\boldsymbol{1}_q$, when it exists, is an element $\sup\mu\in X$ with $|\sup\mu|=q$, such that
$$1_X^{\nat}(\sup\mu,-)=1_X^{\nat}\lda\mu.$$
Dually, the \emph{infimum} of a potential upper $\FQ$-subset $\lam\colon\boldsymbol{1}_q\rto X$, when it exists, is an element $\inf\lam\in X$ with ${|\inf\lam|=q}$, such that
$$1_X^{\nat}(-,\inf\lam)=\lam\rda 1_X^{\nat}.$$
\end{defn}

To explain the above definition in order-theoretic terms, we note from Proposition~\ref{QFRel-comp-imp} that if $\mu\in\PX$, then $\sup\mu$ satisfies
\begin{equation} \label{sup-trans}
1_X^{\nat}(\sup\mu,x)=\bw_{x'\in X}1_X^{\nat}(x',x)\lda\mu(x')
\end{equation}
for all $x\in X$, where $\mu(x')\colon\boldsymbol{1}_{|x'|}\rto\boldsymbol{1}_q$ and $1_X^{\nat}(x',x)\colon\boldsymbol{1}_{|x'|}\rto\boldsymbol{1}_{|x|}$ are considered as $\FQ$-relations between singleton $\FQ$-subsets. Thus \eqref{sup-trans} illustrates the many-valued version of ``$\sup\mu\leq x$ if, and only if, every $x'$ in $(X,\mu)$ satisfies $x'\leq x$'', and $|\sup\mu|=q=|\mu|$ indicates that the degree of $(X,\mu)$ being a lower $\FQ$-subset of $X$ equals to the degree of its supremum in $(X,|\text{-}|)$, whenever it exists.

It is clear that the supremum of $\mu\in\PX$, when it exists, is unique up to isomorphism; that is, if $s,s'\in X$ are both suprema of $\mu$, then $s\cong s'$ in the underlying preorder of $X$. If $X$ is separated, then each $\mu\in\PX$ has at most one supremum. The same facts hold for the infimum of $\lam\in\PdX$.

\begin{prop} {\rm(See \cite{Stubbe2005}.)} \label{sup-inf}
For any $X\in\QFOrd$, each potential lower $\FQ$-subset of $X$ has a supremum if, and only if, each potential upper $\FQ$-subset of $X$ has an infimum.
\end{prop}

\begin{proof}
For any $\mu\in\PX$, it is straightforward to check that $\ub\mu:=1_X^{\nat}\lda\mu\in\PdX$, and $\sup\mu=\inf\ub\mu$ whenever it exists. This proves the ``if'' part, and the ``only if'' part is obtained dually.
\end{proof}

\begin{defn}
A $\FQ$-preordered $\FQ$-subset $X$ is \emph{complete} if each potential lower $\FQ$-subset has a supremum; or equivalently, if each potential upper $\FQ$-subset has an infimum.
\end{defn}

If $X$ is a complete $\FQ$-preordered $\FQ$-subset, then
$$1_{\PX}^{\nat}(\mu,\mu')=\mu'\lda\mu\leq(1_X^{\nat}\lda\mu')\rda(1_X^{\nat}\lda\mu)=1_X^{\nat}(\sup\mu',-)\rda 1_X^{\nat}(\sup\mu,-)=1_X^{\nat}(\sup\mu,\sup\mu')$$
for all $\mu,\mu'\in\PX$, where the last equality follows from the Yoneda lemma; that is to say,
$$\sup\colon\PX\to X$$
is a $\FQ$-order-preserving map. Similarly, it is straightforward to check that so is $\inf\colon\PdX\to X$.

Each $\FQ$-order-preserving map $f\colon X\to Y$ induces $\FQ$-order-preserving maps
$$f^{\ra}\colon\PX\to\PY\quad\text{and}\quad f^{\nra}\colon\PdX\to\PdY$$
with
\begin{equation} \label{fra-def}
f^{\ra}\mu=\mu\comp f^{\nat}\quad\text{and}\quad f^{\nra}\lam=f_{\nat}\comp\lam
\end{equation}
for all $\mu\in\PX$ and $\lam\in\PdX$. $f$ is said to be \emph{$\sup$-preserving} (resp.\ \emph{$\inf$-preserving}) if
$$f{\sup}_X\mu={\sup}_Y f^{\ra}\mu\quad(\text{resp.}\ f{\inf}_X\lam)={\inf}_Y f^{\nra}\lam$$
whenever $\sup_X\mu$ (resp.\ $\inf_X\lam$) exists in $X$.

Separated complete $\FQ$-preordered $\FQ$-subsets and $\sup$-preserving $\FQ$-order-preserving maps constitute an ordered category $\QFSup$, which is a subcategory of $\QFOrd$ and, moreover, is a quantaloid (see the last paragraph of Section~\ref{Fuzzy_Relations}).

Recall that, an object $Z$ in a category $\CC$ is \emph{$\CM$-injective} \cite{Adamek1990,Hofmann2014} w.r.t. a class $\CM$ of morphisms in $\CC$ if, for any morphisms $m\colon X\to Y$ in $\CM$ and $f\colon X\to Z$, there exists a morphism $g\colon Y\to Z$ extending $f$, i.e., making the diagram
$$\bfig
\qtriangle<600,400>[X`Y`Z;m`f`g]
\efig$$
commutative. The following theorem shows that objects in $\QFSup$, i.e., separated complete $\FQ$-preordered $\FQ$-subsets, are characterized as injective objects in $\QFOrd$:

\begin{thm} {\rm(See \cite{Shen2016,Stubbe2017}.)}
A separated $\FQ$-preordered $\FQ$-subset is complete if, and only if, it is injective w.r.t. fully faithful $\FQ$-order-preserving maps.
\end{thm}

\begin{proof}
Let $Z$ be a separated $\FQ$-preordered $\FQ$-subset. If $Z$ is complete, for morphisms $m\colon X\to Y$ and $f\colon X\to Z$ in $\QFOrd$ with $m$ fully faithful,
$$g:=(Y\to^{\widetilde{m_{\nat}}}\PX\to^{f^{\ra}}\PZ\to^{\sup_Z}Z)$$
defines the required extension of $f$, where $\widetilde{m_{\nat}}\colon Y\to\PX$ is given by $\widetilde{m_{\nat}}y=m_{\nat}(-,y)\in\PX$ for all $y\in Y$.

Conversely, one applies the injectivity of $Z$ to the fully faithful Yoneda embedding $\sy_Z\colon Z\to\PZ$, and the resulting extension of $1_Z\colon Z\to Z$ along $\sy_Z$ gives the required $\sup\colon\PZ\to Z$.
\end{proof}

\begin{rem} \label{complete-Q-lattice}
An injective object in $\QOrd$ w.r.t. fully faithful $\FQ$-order-preserving maps is known as a \emph{complete} {\protect\linebreak}\emph{$\FQ$-lattice} \cite{Shen2013}; that is, a (crisp) set $X$ equipped with a separated and complete $\FQ$-preorder. Explicitly, for any ${X\in\QOrd}$ (i.e., $X\in\QFOrd$ with $|x|=e$ for all $x\in X$), a \emph{lower} (resp.\ \emph{upper}) \emph{$\FQ$-subset} of $X$ is precisely a potential lower (resp.\ upper) $\FQ$-subset
$$\mu\colon X\rto\boldsymbol{1}_e\quad(\text{resp}.\ \lam\colon\boldsymbol{1}_e\rto X).$$
$X$ is called a \emph{complete $\FQ$-lattice} if every lower $\FQ$-subset of $X$ admits a supremum, or equivalently, every upper $\FQ$-subset of $X$ admits an infimum, with suprema and infima defined in the same way as in $\QFOrd$.

Although $\QOrd$ is a coreflective subcategory of $\QFOrd$ and the coreflector sends each separated complete {\protect\linebreak}$\FQ$-preordered $\FQ$-subset to a complete $\FQ$-lattice, it is important to notice that a complete $\FQ$-lattice can never be complete as a $\FQ$-preordered $\FQ$-subset as long as $\FQ$ is non-trivial; we will explain it later in Remark~\ref{QFSup-QSup}.
\end{rem}

\subsection{{\textup(}Co{\textup)}tensored $\FQ$-preordered $\FQ$-subsets}

We now introduce tensors and cotensors as a useful tool to characterize complete $\FQ$-preordered $\FQ$-subsets. Recall from Example~\ref{Q-relation-exmp}\,\ref{Q-relation-exmp:DQ} that for any $p,q\in\FQ$, an element $u\in\DQ(p,q)$ may be identified with a $\FQ$-relation ${u\colon\boldsymbol{1}_p\rto\boldsymbol{1}_q}$. Thus we have the following definition:

\begin{defn} \label{tensor-def}
Let $X$ be a $\FQ$-preordered $\FQ$-subset, $x\in X$ and $q\in\FQ$. For any $u\in\DQ(|x|,q)$, the \emph{tensor} of $u$ and $x$, when it exists, is an element $u\otimes x\in X$ with $|u\otimes x|=q$ and
$$1_X^{\nat}(u\otimes x,-)=1_X^{\nat}(x,-)\lda u.$$
Dually, for any $v\in\DQ(|x|,q)$, the \emph{cotensor} of $v$ and $x$, when it exists, is an element $v\rat x\in X$ with $|v\rat x|=q$ and
$$1_X^{\nat}(-,v\rat x)=v\rda 1_X^{\nat}(-,x).$$
$X$ is said to be \emph{tensored} if $u\otimes x$ exists for all $x\in X$, $q\in\FQ$ and $u\in\DQ(|x|,q)$. Dually, $X$ is said to be \emph{cotensored} if $v\rat x$ exists for all $x\in X$, $q\in\FQ$ and $v\in\DQ(q,|x|)$.
\end{defn}

A $\FQ$-preordered $\FQ$-subset $X$ is \emph{order-complete} if, for any $q\in\FQ$, the (crisp) subset
$$X_q=\{x\in X\mid |x|=q\}$$
of $X$ admits all joins (or equivalently, all meets) in the underlying preorder of $X$.

\begin{thm} {\rm(See \cite{Stubbe2006}.)} \label{complete-tensor-order}
A $\FQ$-preordered $\FQ$-subset is complete if, and only if, it is tensored, cotensored, and order-complete.
\end{thm}

\begin{proof}
Let $X$ be a $\FQ$-preordered $\FQ$-subset. For the ``only if'' part, note that for all $x\in X$, $q\in\FQ$, $u\in\DQ(|x|,q)$ and $v\in\DQ(q,|x|)$, the compositions
$$\bfig
\morphism<400,0>[X`\boldsymbol{1}_{|x|};\sy x]
\morphism(400,0)<400,0>[\boldsymbol{1}_{|x|}`\boldsymbol{1}_q;u]
\place(200,0)[\mapstochar] \place(600,0)[\mapstochar]
\efig
\quad\text{and}\quad
\bfig
\morphism(1600,0)<400,0>[\boldsymbol{1}_q`\boldsymbol{1}_{|x|};v]
\morphism(2000,0)<400,0>[\boldsymbol{1}_{|x|}`X;\syd x]
\place(1800,0)[\mapstochar] \place(2200,0)[\mapstochar]
\efig$$
are respectively in $\PX$ and $\PdX$, with
$$u\otimes x=\sup(u\comp\sy x)\quad\text{and}\quad v\rat x=\inf(\syd x\comp v).$$
Similarly, for all $\{x_i\}_{i\in I}\subseteq X_q$, one has $\bv\limits_{i\in I}\sy x_i\in\PX$ and $\bv\limits_{i\in I}x_i=\sup(\bv\limits_{i\in I}\sy x_i)$. Thus $X$ is tensored, cotensored, and order-complete provided that $X$ is complete.

Conversely, the ``if'' part holds since for all $\mu\in\PX$ and $\lam\in\PdX$, one has
\[\sup\mu=\bv\limits_{x\in X}\mu(x)\otimes x\quad\text{and}\quad \inf\lam=\bw\limits_{x\in X}\lam(x)\rat x.\qedhere\]
\end{proof}

\begin{rem} \label{QFSup-QSup}
As an immediate consequence of Theorem~\ref{complete-tensor-order}, one sees that a complete $\FQ$-preordered $\FQ$-subset $X$ must contain at least one element of membership degree $q$ for each $q\in\FQ$, i.e., the bottom element in the underlying preorder of each $X_q$ ($q\in\FQ$) as the join of the empty set. Therefore, provided that $\FQ$ is a non-trivial quantale, a complete $\FQ$-lattice (see Remark~\ref{complete-Q-lattice}) $X$ can never be an object of $\QFSup$ since $|x|=e$ for all $x\in X$.

In the particular case of $\FQ=\boldsymbol{2}$, a complete lattice is \emph{not} complete as a partially defined preordered set (see Example~\ref{QFOrd-exmp}\,\ref{QFOrd-exmp:POrd}). Indeed, a partially defined preordered set $X$ is complete if, and only if, $X\setminus\supp X\neq\varnothing$ and $\supp X$ admits all joins (or equivalently, all meets).
\end{rem}

\begin{exmps} (See \cite{Stubbe2006}.) \label{PX-tensor}
For each $X\in\QFOrd$, $\PX$ and $\PdX$ are both separated, tensored, cotensored and complete $\FQ$-preordered $\FQ$-subsets:
\begin{enumerate}[label=(\arabic*)]
\item Tensors and cotensors in $\PX$ are given by
    $$u\otimes\mu=u\comp\mu\quad\text{and}\quad v\rat\mu=v\rda\mu$$
    for all $\mu\in\PX$, $q\in\FQ$, $u\in\DQ(|\mu|,q)$ and $v\in\DQ(q,|\mu|)$, and consequently the Yoneda lemma implies
    \begin{align*}
    \sup\Theta&=\bv_{\mu\in\PX}\Theta(\mu)\comp\mu=\bv_{\mu\in\PX}\Theta(\mu)\comp(\sy_X)_{\nat}(-,\mu)=\Theta\comp(\sy_X)_{\nat} \quad\text{and}\\
    \inf\Lam&=\bw_{\mu\in\PX}\Lam(\mu)\rda\mu=\bw_{\mu\in\PX}\Lam(\mu)\rda(\sy_X)_{\nat}(-,\mu)=\Lam\rda(\sy_X)_{\nat}
    \end{align*}
    for all $\Theta\in\sP(\PX)$ and $\Lam\in\sPd(\PX)$.
\item Tensors and cotensors in $\PdX$ are given by
    $$u\otimes\lam=\lam\lda u,\quad v\rat\lam=\lam\comp v$$
    for all $\lam\in\PdX$, $q\in\FQ$, $u\in\DQ(|\lam|,q)$ and $v\in\DQ(q,|\lam|)$, and consequently the Yoneda lemma implies
    \begin{align*}
    \sup\Theta&=\bw_{\lam\in\PdX}\lam\lda\Theta(\lam)=\bw_{\lam\in\PdX}(\syd_X)^{\nat}(\lam,-)\lda\Theta(\lam)=(\syd_X)^{\nat}\lda\Theta \quad\text{and}\\
    \inf\Lam&=\bv_{\lam\in\PdX}\lam\comp\Lam(\lam)=\bv_{\lam\in\PdX}(\syd_X)^{\nat}(\lam,-)\comp\Lam(\lam)=(\syd_X)^{\nat}\comp\Lam
    \end{align*}
    for all $\Theta\in\sP(\PdX)$ and $\Lam\in\sPd(\PdX)$.
\end{enumerate}
\end{exmps}

\section{Fuzzy Galois connections on fuzzy sets} \label{Fuzzy-Galois-connections}

\subsection{$\FQ$-distributors}

While dealing with $\FQ$-relations between $\FQ$-preordered $\FQ$-subsets, it is natural to consider those $\phi\colon X\rto Y$ which are compatible with the $\FQ$-preorder structures on $X$ and $Y$; such $\FQ$-relations are called $\FQ$-distributors:

\begin{defn}
A \emph{$\FQ$-distributor} $\phi\colon X\oto Y$ between $\FQ$-preordered $\FQ$-subsets is a $\FQ$-relation $\phi\colon X\rto Y$ satisfying
\begin{equation} \label{QDist-def}
1_Y^{\nat}\comp\phi\comp 1_X^{\nat}\leq\phi.
\end{equation}
\end{defn}

Since the reverse inequality of \eqref{QDist-def} is trivial, a $\FQ$-distributor $\phi\colon X\oto Y$ necessarily satisfies
$$1_Y^{\nat}\comp\phi\comp 1_X^{\nat}=\phi.$$
In fact, there are more equivalent ways of describing the ``compatibility'' of a $\FQ$-relation with the $\FQ$-preorder on its domain and codomain:

\begin{prop}
For a $\FQ$-relation $\phi\colon X\rto Y$ between $\FQ$-preordered $\FQ$-subsets, the following statements are equivalent{\textup:}
\begin{enumerate}[label={\rm(\roman*)}]
\item $\phi\colon X\oto Y$ is a $\FQ$-distributor.
\item $(1_Y^{\nat}(y,y')\ldd|y|)\with\phi(x,y)\with(|x|\rdd 1_X^{\nat}(x',x))\leq\phi(x',y')$ for all $x,x'\in X$ and $y,y'\in Y$.
\item $\phi\comp 1_X^{\nat}\leq\phi$ and $1_Y^{\nat}\comp\phi\leq\phi$.
\item $\phi(x,-)\in\PdY$ and $\phi(-,y)\in\PX$ for all $x\in X$ and $y\in Y$.
\item $1_X^{\nat}\leq\phi\rda\phi$ and $1_Y^{\nat}\leq\phi\lda\phi$.
\end{enumerate}
\end{prop}

\begin{exmps} \
\begin{enumerate}[label=(\arabic*)]
\item For any $\FQ$-order-preserving map $f\colon X\to Y$, its graph $f_{\nat}\colon X\oto Y$ and cograph $f^{\nat}\colon Y\oto X$ are both {\protect\linebreak}$\FQ$-distributors.
\item Potential lower (resp.\ upper) $\FQ$-subsets of $X$ are precisely $\FQ$-distributors
$$\mu\colon X\oto\boldsymbol{1}_q\quad(\text{resp}.\ \lam\colon\boldsymbol{1}_q\oto X).$$
\end{enumerate}
\end{exmps}

$\FQ$-preordered $\FQ$-subsets and $\FQ$-distributors constitute a quantaloid $\QFDist$, which contains $\QFRel$ as a full subquantaloid; indeed, every $\FQ$-relation $\phi\colon X\rto Y$ between $\FQ$-subsets can be regarded as a $\FQ$-distributor between $X$ and $Y$ equipped with the discrete $\FQ$-preorder. Compositions and implications of $\FQ$-distributors are calculated in the same way as those of $\FQ$-relations (see Equations~\eqref{psi-circ-phi} and Proposition~\ref{QFRel-comp-imp}), while the identity $\FQ$-distributor on $X\in\QFOrd$ is given by its $\FQ$-preorder $1_X^{\nat}\colon X\oto X$.

\begin{defn}
A pair of $\FQ$-distributors $\phi\colon X\oto Y$ and $\psi\colon Y\oto X$ forms an \emph{adjunction} in $\QFDist$, written as $\phi\dv\psi$, if
$$1_X^{\nat}\leq\psi\comp\phi\quad\text{and}\quad \phi\comp\psi\leq 1_Y^{\nat}.$$
In this case, one says that $\phi$ is a \emph{left adjoint} of $\psi$, and $\psi$ is a \emph{right adjoint} of $\phi$.
\end{defn}

Using the language of category theory, adjoint $\FQ$-distributors are in fact \emph{internal adjunctions} in the ordered category $\QFDist$.

\begin{exmps} \label{adjoint-dist-exmp} \
\begin{enumerate}[label=(\arabic*)]
\item \label{adjoint-dist-exmp:graph}
    Every $\FQ$-order-preserving map $f\colon X\to Y$ induces an adjunction $f_{\nat}\dv f^{\nat}$ in $\QFDist$.
\item \label{adjoint-dist-exmp:POrd}
 Let $\phi\colon X\oto Y$ and $\psi\colon Y\oto X$ be a pair of distributors between partially defined preordered sets (see Example~\ref{QFOrd-exmp}\,\ref{QFOrd-exmp:POrd}). Then $\phi\dv\psi$ if, and only if, $\phi=f_{\nat}$ and $\psi=f^{\nat}$ for some order-preserving map ${f\colon\supp X\to\supp Y}$.
\item \label{adjoint-dist-exmp:skew}
    For a frame $\Om$, in Example~\ref{QFOrd-exmp}\,\ref{QFOrd-exmp:skew} we have seen that $\Om$-preordered $\Om$-subsets are skew $\Om$-sets. The category
    $$\Om_{\sf sk}\text{-}\Set$$ of skew $\Om$-sets and their morphisms given in \cite{Borceux1998} is precisely a subcategory of $\Om\text{-}\FDist$: its objects are also $\Om$-preordered $\Om$-subsets, while its morphisms are left adjoint $\Om$-distributors. In other words, $\phi\colon X\oto Y$ is a morphism in $\Om_{\sf sk}\text{-}\Set$ if there exists $\psi\colon Y\oto X$ such that $\phi\dv\psi$ in $\Om\text{-}\FDist$.
\item \label{adjoint-dist-exmp:Cauchy}
    Let $\FQ=([0,\infty]^{\op},+,0)$ and $(X,\al)$ be a partial metric space (see Example~\ref{QFOrd-exmp}\,\ref{QFOrd-exmp:ParMet}). A sequence $\{x_n\}_{n=1}^{\infty}\subseteq X$ is \emph{Cauchy} if the limit $\lim\limits_{n,m\ra\infty}\al(x_n,x_m)$ exists in $[0,\infty]$, and a Cauchy sequence $\{x_n\}_{n=1}^{\infty}\subseteq X$ \emph{converges} to $x\in X$ \cite{Kuenzi2006a} if
    $$\al(x,x)=\lim\limits_{n\ra\infty}\al(x,x_n)=\lim\limits_{n\ra\infty}\al(x_n,x)=\lim\limits_{n,m\ra\infty}\al(x_n,x_m).$$
    Then every Cauchy sequence $\{x_n\}_{n=1}^{\infty}\subseteq X$ induces an adjunction $\lam\dv\mu$ in $\QFDist$ with
    \begin{align*}
    &\lam\colon\boldsymbol{1}_q\oto X,\quad\lam(y)=\lim\limits_{n\ra\infty}\al(x_n,y)\quad\text{and}\quad\mu\colon X\oto\boldsymbol{1}_q,\quad\mu(y)=\lim\limits_{n\ra\infty}\al(y,x_n)
    \end{align*}
    for all $y\in X$, where $q=\lim\limits_{n,m\ra\infty}\al(x_n,x_m)$. It can be shown that $\{x_n\}_{n=1}^{\infty}$ converges to $x\in X$ if, and only if, $\lam(y)=\al(x,y)$ and $\mu(y)=\al(y,x)$ for all $y\in X$ (see \cite[Proposition~4.10]{Pu2012}); that is, $\lam$ and $\mu$ are respectively the graph and cograph of the (necessarily non-expanding) map
    $$f\colon\boldsymbol{1}_q\to X,\quad *\longmapsto x.$$
\end{enumerate}
\end{exmps}

The identities presented below are quite useful when being applied to the adjunction $f_{\nat}\dv f^{\nat}$:

\begin{prop} {\rm(See \cite{Heymans2010}.)} \label{adj-dist-cal}
If $\phi\dv\psi$ in $\QFDist$, then the following identities hold for all $\FQ$-distributors $\xi$ and $\xi'$ whenever the compositions and implications make sense{\textup:}
\begin{enumerate}[label={\rm (\arabic*)}]
\item \label{adj-dist-cal:comp}
    $\xi\comp\phi=\xi\lda\psi$ and $\psi\comp\xi=\phi\rda\xi$.
\item \label{adj-dist-cal:comp-imp}
    $(\phi\comp\xi)\rda\xi'=\xi\rda(\psi\comp\xi')$ and $(\xi'\comp\phi)\lda\xi=\xi'\lda(\xi\comp\psi)$.
\item \label{adj-dist-cal:imp-comp}
    $(\xi\rda\xi')\comp\phi=\xi\rda(\xi'\comp\phi)$ and $\psi\comp(\xi'\lda\xi)=(\psi\comp\xi')\lda\xi$.
\item \label{adj-dist-cal:order}
    $\psi\comp(\xi\rda\xi')=(\xi\comp\phi)\rda\xi'$ and $(\xi'\lda\xi)\comp\phi=\xi'\lda(\psi\comp\xi)$.
\end{enumerate}
\end{prop}

\subsection{$\FQ$-Galois connections}

Parallel to the definition of adjoint $\FQ$-distributors, internal adjunctions in the ordered category $\QFOrd$ give the definition of $\FQ$-Galois connections between $\FQ$-preordered $\FQ$-subsets:

\begin{defn}
A pair of $\FQ$-order-preserving maps $f\colon X\to Y$ and $g\colon Y\to X$ forms a \emph{$\FQ$-Galois connection} (or, a \emph{$\FQ$-adjunction}), written as $f\dv g$, if
$$1_X\leq gf\quad\text{and}\quad fg\leq 1_Y.$$
In this case, one says that $f$ is a \emph{left adjoint} of $g$, and $g$ is a \emph{right adjoint} of $f$.
\end{defn}

\begin{rem}
If one considers $\FQ$-preordered $\FQ$-subsets as categories enriched in the quantaloid $\DQ$ (see Remark~\ref{QFOrd-CT}), then $\FQ$-Galois connections are precisely \emph{adjoint $\DQ$-functors} between $\DQ$-enriched categories.
\end{rem}

It is useful to characterize $\FQ$-Galois connections in the following ways:

\begin{prop} {\rm(See \cite{Stubbe2005}.)} \label{adjoint-graph}
Let $f\colon X\to Y$ and $g\colon Y\to X$ be a pair of $\FQ$-order-preserving maps. Then the following statements are equivalent{\textup:}
\begin{enumerate}[label={\rm(\roman*)}]
\item $f\dv g$ in $\QFOrd$.
\item \label{adjoint-graph:fnat=gnat}
    $f_{\nat}=g^{\nat}${\textup;} that is, $1_Y^{\nat}(fx,y)=1_X^{\nat}(x,gy)$ for all $x\in X$ and $y\in Y$.
\item $g_{\nat}\dv f_{\nat}$ in $\QFDist$.
\item $g^{\nat}\dv f^{\nat}$ in $\QFDist$.
\end{enumerate}
\end{prop}

Condition~\ref{adjoint-graph:fnat=gnat} in the above proposition is in fact strong enough to determine a $\FQ$-Galois connection even without the premise that $f$ and $g$ are $\FQ$-order-preserving:

\begin{prop} {\rm(See \cite{Shen2014}.)}
If $f\colon X\to Y$ and $g\colon Y\to X$ are a pair of membership-preserving maps between $\FQ$-preordered $\FQ$-subsets (need not be $\FQ$-order-preserving maps), then the following statements are equivalent{\textup:}
\begin{enumerate}[label={\rm(\roman*)}]
\item $f$ and $g$ are $\FQ$-order-preserving, and $f\dv g$ in $\QFOrd$.
\item $1_Y^{\nat}(fx,y)=1_X^{\nat}(x,gy)$ for all $x\in X$ and $y\in Y$.
\end{enumerate}
\end{prop}

From the above characterizations of $\FQ$-Galois connections one easily sees that the right adjoint (or the left adjoint) of $f\colon X\to Y$ in $\QFOrd$, when it exists, is unique up to isomorphism; that is, if $g,g'\colon Y\to X$ are both right adjoints (or both left adjoints) of $f$, then $g\cong g'$ in the ordered hom-set $\QFOrd(Y,X)$, and one has $g=g'$ when $X$ is separated.

\begin{exmps} \label{Q-Galois-exmp} \
\begin{enumerate}[label={\rm (\arabic*)}]
\item \label{Q-Galois-exmp:ub-lb}
    Each $X\in\QFOrd$ induces a $\FQ$-Galois connection
    $\bfig
    \morphism/@{->}@<4pt>/<500,0>[\PX`\PdX;\ub]
    \morphism(500,0)|b|/@{->}@<4pt>/<-500,0>[\PdX`\PX;\lb]
    \place(240,5)[\mbox{\footnotesize{$\bot$}}]
    \efig$,
    where $\ub\colon\PX\to\PdX$ (see the proof of Proposition~\ref{sup-inf}) and $\lb\colon\PdX\to\PX$ are respectively given by
    $$\ub\mu=1_X^{\nat}\lda\mu\quad\text{and}\quad\lb\lam=\lam\rda 1_X^{\nat}$$
    for all $\mu\in\PX$ and $\lam\in\PdX$.
 \item \label{Q-Galois-exmp:image}
    Every $\FQ$-order-preserving map $f\colon X\to Y$ gives rise to two $\FQ$-Galois connections
    $$\bfig
    \morphism/@{->}@<4pt>/<500,0>[\PX`\PY;f^{\ra}]
    \morphism(500,0)|b|/@{->}@<4pt>/<-500,0>[\PY`\PX;f^{\la}]
    \place(250,5)[\mbox{\footnotesize{$\bot$}}]
    \efig
    \quad\text{and}\quad
    \bfig
    \morphism/@{->}@<4pt>/<500,0>[\PdY`\PdX,;f^{\nla}]
    \morphism(500,0)|b|/@{->}@<4pt>/<-500,0>[\PdX,`\PdY;f^{\nra}]
    \place(250,5)[\mbox{\footnotesize{$\bot$}}]
    \efig$$
    where $f^{\ra}$ and $f^{\nra}$ are defined as in \eqref{fra-def}, while
    $$f^{\la}\mu=\mu\comp f_{\nat}\quad\text{and}\quad f^{\nla}\lam=f^{\nat}\comp\lam$$
    for all $\mu\in\PY$ and $\lam\in\PdY$.
\item \label{Q-Galois-exmp:complete}
    For any $X\in\FOrd$, the following statements are equivalent:
    \begin{enumerate}[label=(\roman*)]
    \item $X$ is complete;
    \item the Yoneda embedding $\sy\colon X\to\PX$ has a left adjoint in $\QFOrd$, given by $\sup\colon\PX\to X$;
    \item the co-Yoneda embedding $\syd\colon X\to\PdX$ has a right adjoint in $\QFOrd$, given by $\inf\colon\PdX\to X$.
    \end{enumerate}
\item \label{Q-Galois-exmp:tensor}
    For any $p\in\FQ$, from Example~\ref{PX-exmp}\,\ref{PX-exmp:singleton} we see that
    $$\sP\boldsymbol{1}_p=\{u\in\DQ(p,q)\mid q\in\FQ\}\quad\text{and}\quad\sPd\boldsymbol{1}_p=\{v\in\DQ(q,p)\mid q\in\FQ\}.$$
    Let $X\in\QFOrd$. Then for any $x\in X$,
    $$1_X^{\nat}(x,-)\colon X\to\sP\boldsymbol{1}_{|x|},\quad y\longmapsto 1_X^{\nat}(x,y)\quad\text{and}\quad 1_X^{\nat}(-,x)\colon X\to\sPd\boldsymbol{1}_{|x|},\quad y\longmapsto 1_X^{\nat}(y,x)$$
    are both $\FQ$-order-preserving maps. It follows soon from Definition~\ref{tensor-def} that
    \begin{itemize}
    \item $X$ is tensored if, and only if, for every $x\in X$, $1_X^{\nat}(x,-)\colon X\to\sP\boldsymbol{1}_{|x|}$ has a left adjoint in $\QFOrd$, given by $-\otimes x\colon\sP\boldsymbol{1}_{|x|}\to X$;
    \item $X$ is cotensored if, and only if, for every $x\in X$, $1_X^{\nat}(-,x)\colon X\to\sPd\boldsymbol{1}_{|x|}$ has a right adjoint in $\QFOrd$, given by $-\rat x\colon\sPd\boldsymbol{1}_{|x|}\to X$.
    \end{itemize}
\item \label{Q-Galois-exmp:Cauchy}
    A partial metric space $X=(X,\al)$ is \emph{Cauchy complete} \cite{Pu2012} if every Cauchy sequence in $X$ converges (see Example~\ref{adjoint-dist-exmp}\,\ref{adjoint-dist-exmp:Cauchy}) and there exists $x\in X$ with $\al(x,x)=\infty$. Considering $X$ as a $\FQ$-preordered $\FQ$-subset, where $\FQ=([0,\infty]^{\op},+,0)$, it makes sense to define
    \begin{equation} \label{PXc-def}
    (\PX)_{\sf c}:=\{\mu\colon X\oto\boldsymbol{1}_q\mid\mu\ \text{is a right adjoint in}\ \QFDist,\ q\in\FQ\},
    \end{equation}
    which becomes a $\FQ$-preordered $\FQ$-subset with structures inherited from $\PX$. Then the Cauchy completeness of partial metric spaces can be characterized by $\FQ$-Galois connections: a partial metric space $X$ is Cauchy complete if, and only if, the restriction of the Yoneda embedding $\sy\colon X\to(\PX)_{\sf c}$ has a left adjoint in $\QFOrd$, given by $\sup\colon(\PX)_{\sf c}\to X$.

  For a general $\FQ$, one may define the Cauchy completeness of $X\in\QFOrd$ in the same way: $X$ is \emph{Cauchy complete} if $\sy\colon X\to(\PX)_{\sf c}$ has a left adjoint, where $(\PX)_{\sf c}$ is given by \eqref{PXc-def}. When $\FQ=\Om$ is a frame, Cauchy complete $\Om$-preordered $\Om$-subsets are exactly \emph{complete skew $\Om$-sets} in the sense of Borceux-Cruciani \cite{Borceux1998}.
\end{enumerate}
\end{exmps}

Moreover, left adjoint maps between complete $\FQ$-preordered $\FQ$-subsets are precisely $\sup$-preserving maps:

\begin{thm} {\rm(See \cite{Stubbe2005,Stubbe2006}.)} \label{la-condition}
For any $\FQ$-order-preserving map $f\colon X\to Y$, with $X$ complete, the following statements are equivalent{\textup:}
\begin{enumerate}[label={\rm (\roman*)}]
\item \label{la-condition:la}
    $f$ is a left {\textup(}resp.\ right{\textup)} adjoint in $\QFOrd$.
\item \label{la-condition:sup}
    $f$ is $\sup$-preserving {\textup(}resp.\ $\inf$-preserving{\textup)}.
\item \label{la-condition:tensor}
    $f$ is a left {\textup(}resp.\ right{\textup)} adjoint between the underlying preordered sets of $X$ and $Y$, and preserves tensors {\textup(}resp.\ cotensors{\textup)} in the sense that $f(u\otimes_X x)= u\otimes_Y fx$ {\textup(}resp.\ $f(v\rat_X x)=v\rat_Y fx${\textup)} for all $x\in X$, $q\in\FQ$ and $u\in\DQ(|x|,q)$ {\textup(}resp.\ $v\in\DQ(q,|x|)${\textup)}.
\end{enumerate}
\end{thm}

\begin{proof}
\ref{la-condition:la}$\implies$\ref{la-condition:sup}: If $g\colon Y\to X$ is a right adjoint of $f$, then for all $\mu\in\PX$,
$$1_Y^{\nat}(f{\sup}_X\mu,-)=1_X^{\nat}({\sup}_X\mu,g-)=1_X^{\nat}(-,g-)\lda\mu=f_{\nat}\lda\mu=1_Y^{\nat}\lda(\mu\comp f^{\nat})=1_Y^{\nat}\lda f^{\ra}\mu,$$
where the penultimate equality follows from Proposition~\ref{adj-dist-cal}\,\ref{adj-dist-cal:comp-imp} and the fact that $f_{\nat}\dv f^{\nat}$. Thus $f\sup_X\mu=\sup_Y f^{\ra}\mu$.

\ref{la-condition:sup}$\implies$\ref{la-condition:tensor}: Since tensors and underlying joins are both suprema (see the proof of Theorem~\ref{complete-tensor-order}), the conclusion soon follows.

\ref{la-condition:tensor}$\implies$\ref{la-condition:la}: Since $f$ has a right adjoint $g\colon Y\to X$ in the underlying preorder, it suffices to show that $f\dv g$ in $\QFOrd$. To this end, we show that $1_Y^{\nat}(fx,y)=1_X^{\nat}(x,gy)$ for all $x\in X$ and $y\in Y$. On one hand, ${1_X^{\nat}(x,gy)\leq 1_Y^{\nat}(fx,y)}$ since
$$|y|\leq 1_X^{\nat}(x,gy)\lda 1_X^{\nat}(x,gy)=1_X^{\nat}\bigl(\bigl(1_X^{\nat}(x,gy)\mathrel{\otimes_X}x\bigr),gy\bigr)$$
implies
$$|y|\leq 1_Y^{\nat}\bigl(f\bigl(1_X^{\nat}(x,gy)\mathrel{\otimes_X} x\bigl),y\bigr)=1_Y^{\nat}\bigl(\bigl(1_X^{\nat}(x,gy)\mathrel{\otimes_Y} fx\bigr),y\bigr)=1_Y^{\nat}(fx,y)\lda 1_X^{\nat}(x,gy),$$
following the facts that $f\dv g$ in $\Ord$ and $f$ preserves tensors. On the other hand, $1_Y^{\nat}(fx,y)\leq 1_X^{\nat}(x,gy)$ can be checked similarly, which completes the proof.
\end{proof}

\subsection{$\FQ$-polarities and {\textup(}dual{\textup)} $\FQ$-axialities}

In this subsection we are concerned with $\FQ$-Galois connections between (dual) $\FQ$-powersets of $\FQ$-preordered {\protect\linebreak}$\FQ$-subsets. With a slight modification of the terminologies in \cite{GutierrezGarcia2010}, which originated from \cite{Birkhoff1967,Erne1993a}, we have the following definition:

\begin{defn}
Let $X$ and $Y$ be $\FQ$-preordered $\FQ$-subsets.
\begin{enumerate}[itemsep=-0.7em,label=(\arabic*)]
\item A \emph{$\FQ$-polarity} from $X$ to $Y$ is a $\FQ$-Galois connection
    $\bfig
    \morphism/@{->}@<4pt>/<500,0>[\PX`\PdY;f]
    \morphism(500,0)|b|/@{->}@<4pt>/<-500,0>[\PdY`\PX;g]
    \place(240,5)[\mbox{\footnotesize{$\bot$}}]
    \efig$.
\item A \emph{$\FQ$-axiality} from $X$ to $Y$ is a $\FQ$-Galois connection
    $\bfig
    \morphism/@{->}@<4pt>/<500,0>[\PX`\PY;f]
    \morphism(500,0)|b|/@{->}@<4pt>/<-500,0>[\PY`\PX;g]
    \place(250,5)[\mbox{\footnotesize{$\bot$}}]
    \efig$.
\item A \emph{dual $\FQ$-axiality} from $X$ to $Y$ is a $\FQ$-Galois connection
    $\bfig
    \morphism/@{->}@<4pt>/<500,0>[\PdX`\PdY;f]
    \morphism(500,0)|b|/@{->}@<4pt>/<-500,0>[\PdY`\PdX;g]
    \place(250,5)[\mbox{\footnotesize{$\bot$}}]
    \efig$.
\end{enumerate}
\end{defn}

\begin{rem}
If we extend directly the terminologies of \cite{GutierrezGarcia2010} into the setting of $\FQ$-preordered $\FQ$-subsets, then it is not difficult to observe that
\begin{itemize}
\item a \emph{type I $\FQ$-polarity} from $X$ to $Y$ is a $\FQ$-polarity from $X$ to $Y$;
\item a \emph{type II $\FQ$-polarity} from $X$ to $Y$ is a $\FQ$-polarity from $Y$ to $X$, i.e., a $\FQ$-Galois connection
    $\bfig
    \morphism/@{->}@<4pt>/<500,0>[\PY`\PdX;f]
    \morphism(500,0)|b|/@{->}@<4pt>/<-500,0>[\PdX`\PY;g]
    \place(240,5)[\mbox{\footnotesize{$\bot$}}]
    \efig$;
\item a \emph{type I $\FQ$-axiality} from $X$ to $Y$ is a $\FQ$-axiality from $X$ to $Y$;
\item a \emph{type II $\FQ$-axiality} from $X$ to $Y$ is a dual $\FQ$-axiality from $Y$ to $X$.
\end{itemize}
So, as is already mentioned in \cite[Remark~3.2]{GutierrezGarcia2010} for the case of $\FQ$-preordered sets, type I $\FQ$-polarities from $X$ to $Y$ are precisely type II $\FQ$-polarities from $Y$ to $X$; that is why we combine these two concepts into ``$\FQ$-polarities''. However, the two types of $\FQ$-axialities are essentially different as they cannot be switched to each other simply by swapping the positions of $X$ and $Y$.
\end{rem}


As a special case of Theorem~\ref{la-condition}, we have the following characterizations of (dual) $\FQ$-axialities:

\begin{prop} \label{la-condition-PX}
Let $X,Y\in\QFOrd$ and $f\colon\PX\to\PY$ be a $\FQ$-order-preserving map. Then the following statements are equivalent{\textup:}
\begin{enumerate}[label={\rm(\roman*)}]
\item \label{la-condition-PX:la}
    $f$ is a left adjoint in $\QFOrd${\textup;} that is, there exists $g\colon\PY\to\PX$ such that $f\dv g$ forms a $\FQ$-axiality from $X$ to $Y$.
\item \label{la-condition-PX:tensor}
    $f$ is a left adjoint between the underlying preordered sets of $\PX$ and $\PY$, and $f(u\comp\mu)= u\comp f\mu$ for all $\mu\in\PX$, $q\in\FQ$ and $u\in\DQ(|\mu|,q)$.
\item \label{la-condition-PX:Yoneda}
    $f$ is a left adjoint between the underlying preordered sets of $\PX$ and $\PY$, and $f(u\comp\sy_X x)= u\comp f\sy_X x$ for all $x\in X$, $q\in\FQ$ and $u\in\DQ(|x|,q)$.
\end{enumerate}
\end{prop}

\begin{proof}
\ref{la-condition-PX:la}$\iff$\ref{la-condition-PX:tensor} is an immediate consequence of Theorem~\ref{la-condition} and Example~\ref{PX-tensor}, and \ref{la-condition-PX:tensor}$\implies$\ref{la-condition-PX:Yoneda} is trivial. For \ref{la-condition-PX:Yoneda}$\implies$\ref{la-condition-PX:tensor}, note that for any $\mu\in\PX$ and $u\in\DQ(|\mu|,q)$,
\begin{align*}
f(u\comp\mu)&=f(u\comp\mu\comp 1_X^{\nat})&(\text{Equation~\eqref{PX-def-eq}})\\
&=f\Big(\bv_{x\in X}u\comp\mu(x)\comp 1_X^{\nat}(-,x)\Big)&(\text{Proposition~\ref{QFRel-comp-imp}})\\
&=f\Big(\bv_{x\in X}u\comp\mu(x)\comp\sy_X x\Big)\\
&=\bv_{x\in X}u\comp\mu(x)\comp f\sy_X x&\text{(Condition~\ref{la-condition-PX:Yoneda})}\\
&=u\comp\bv_{x\in X}\mu(x)\comp f\sy_X x&(\text{Proposition~\ref{QFRel-comp}\,\ref{QFRel-comp:join}})\\
&=u\comp f\Big(\bv_{x\in X}\mu(x)\comp\sy_X x\Big)&\text{(Condition~\ref{la-condition-PX:Yoneda})}\\
&=u\comp f\mu,&(\text{Equation~\eqref{PX-def-eq} and Proposition~\ref{QFRel-comp-imp}})
\end{align*}
which completes the proof.
\end{proof}

\begin{prop} \label{la-condition-PdX}
Let $X,Y\in\QFOrd$ and $g\colon\PdY\to\PdX$ be a $\FQ$-order-preserving map. Then the following statements are equivalent{\textup:}
\begin{enumerate}[label={\rm(\roman*)}]
\item \label{la-condition-PdX:la}
    $g$ is a right adjoint in $\QFOrd${\textup;} that is, there exists $f\colon\PdX\to\PdY$ such that $f\dv g$ forms a dual $\FQ$-axiality from $X$ to $Y$.
\item \label{la-condition-PdX:tensor}
    $g$ is a right adjoint between the underlying preordered sets of $\PdY$ and $\PdX$, and $g(\lam\comp v)=g\lam\comp v$ for all $\lam\in\PdY$, $q\in\FQ$ and $v\in\DQ(q,|\lam|)$.
\item \label{la-condition-PdX:Yoneda}
    $g$ is a right adjoint between the underlying preordered sets of $\PdY$ and $\PdX$, and $g(\syd_Y y\comp v)= g\syd_Y y\comp v$ for all $y\in Y$, $q\in\FQ$ and $v\in\DQ(q,|y|)$.
\end{enumerate}
\end{prop}

\begin{proof}
Similar to Proposition~\ref{la-condition-PX}, though one needs to be careful about the underlying preorders of $\PdX$ and $\PdY$ (see Remark~\ref{PdX-order}).
\end{proof}

\subsection{$\FQ$-Galois connections vs. $\FQ$-polarities and {\textup(}dual{\textup)} $\FQ$-axialities}

The following proposition indicates that every $\FQ$-distributor gives rise to a $\FQ$-polarity, a $\FQ$-axiality and a dual $\FQ$-axiality:

\begin{prop} {\rm(See \cite{Lai2017,Shen2014,Shen2013a}.)} \label{Isbell-Kan-def}
Each $\FQ$-distributor $\phi\colon X\oto Y$ between $\FQ$-preordered $\FQ$-subsets induces three $\FQ$-Galois connections between their {\textup(}dual{\textup)} $\FQ$-powersets{\textup:}\footnote{In \cite{Lai2017,Shen2014,Shen2013a}, the $\FQ$-Galois connections $\uphi\dv\dphi$, $\phi^*\dv\phi_*$ and $\phi_{\dag}\dv\phi^{\dag}$ are respectively called the \emph{Isbell adjunction}, \emph{Kan adjunction} and \emph{dual Kan adjunction} induced by a $\FQ$-distributor $\phi$.}
\begin{enumerate}[label={\rm(\arabic*)}]
\item A $\FQ$-polarity $\uphi\dv\dphi$ from $X$ to $Y$ with
\begin{align*}
&\uphi\colon\PX\to\PdY,\quad\mu\longmapsto\phi\lda\mu\quad\text{and}\\
&\dphi\colon\PdY\to\PX,\quad\lam'\longmapsto\lam'\rda\phi
\end{align*}
for all $\mu\in\PX$ and $\lam'\in\PdY$.
\item A $\FQ$-axiality $\phi^*\dv\phi_*$ from $Y$ to $X$ with
\begin{align*}
&\phi^*\colon\PY\to\PX,\quad \mu'\longmapsto\mu'\comp\phi\quad\text{and}\\
&\phi_*\colon\PX\to\PY,\quad \mu\longmapsto\mu\lda\phi
\end{align*}
for all $\mu'\in\PY$ and $\mu\in\PX$.
\item A dual $\FQ$-axiality $\phi_{\dag}\dv\phi^{\dag}$ from $Y$ to $X$ with
\begin{align*}
&\phi_{\dag}\colon\PdY\to\PdX,\quad\lam'\longmapsto\phi\rda\lam'\quad \text{and}\\
&\phi^{\dag}\colon\PdX\to\PdY,\quad\lam\longmapsto\phi\comp\lam
\end{align*}
for all $\lam'\in\PdY$ and $\lam\in\PdX$.
\end{enumerate}
\end{prop}

\begin{proof}
For all $\mu\in\PX$, $\lam\in\PdX$, $\mu'\in\PY$ and $\lam'\in\PdY$, one could easily perform the following calculations using the formulas in Proposition~\ref{QFRel-cal}:
\begin{align*}
1_{\PdY}^{\nat}(\uphi\mu,\lam')&=\lam'\rda(\phi\lda\mu)=(\lam'\rda\phi)\lda\mu=1_{\PX}^{\nat}(\mu,\dphi\lam'),\\
1_{\PX}^{\nat}(\phi^*\mu',\mu)&=\mu\lda(\mu'\comp\phi)=(\mu\lda\phi)\lda\mu'=1_{\PY}^{\nat}(\mu',\phi_*\mu),\\
1_{\PdX}^{\nat}(\phi_{\dag}\lam',\lam)&=\lam\rda(\phi\rda\lam')=(\phi\comp\lam)\rda\lam'=1_{\PdY}^{\nat}(\lam',\phi^{\dag}\lam).\qedhere
\end{align*}
\end{proof}

\begin{exmps} \label{polarities-axialities-exmp} \
\begin{enumerate}[label=(\arabic*)]
\item For any $X\in\QFOrd$, let $\phi=1_X^{\nat}$. Then the $\FQ$-polarity $\bigl(1_X^{\nat}\bigr)_{\ua}\dv\bigl(1_X^{\nat}\bigr)^{\da}$ on $X$ is precisely the $\FQ$-Galois connection $\ub\dv\lb$ given in Example~\ref{Q-Galois-exmp}\,\ref{Q-Galois-exmp:ub-lb}. Obviously, the $\FQ$-axiality $\bigl(1_X^{\nat}\bigr)^*\dv\bigl(1_X^{\nat}\bigr)_*$ and the dual $\FQ$-axiality ${\bigl(1_X^{\nat}\bigr)_{\dag}\dv\bigl(1_X^{\nat}\bigr)^{\dag}}$ on $X$ are both the identity maps on $\PX$ and $\PdX$, respectively.

    It is noteworthy to point out that the fixed points of the $\FQ$-polarity $\bigl(1_X^{\nat}\bigr)_{\ua}\dv\bigl(1_X^{\nat}\bigr)^{\da}$ constitute a complete {\protect\linebreak}$\FQ$-preordered $\FQ$-subset, which is precisely the \emph{MacNeille completion} of $X$ (see \cite[Remark~4.17\,(2)]{Shen2013a} and \cite[Section~5.5]{Shen2014}).
\item A $\FQ$-relation $\phi\colon X\rto Y$ between $\FQ$-subsets is considered as a \emph{fuzzy context} $(X,Y,\phi)$ in formal concept analysis (FCA) and rough set theory (RST) on fuzzy sets \cite{Lai2017,Shen2014,Shen2013b}. Considering $\phi$ as a $\FQ$-distributor between discrete $\FQ$-preordered $\FQ$-subsets, the $\FQ$-polarity $\uphi\dv\dphi$ and the $\FQ$-axiality $\phi^*\dv\phi_*$ induced by $\phi$ are the fundamental operators in FCA and RST, and their fixed points constitute the ``concept lattice'' (which are both complete $\FQ$-preordered $\FQ$-subsets) of the fuzzy context $(X,Y,\phi)$ based on FCA and RST, respectively.
\end{enumerate}
\end{exmps}

Recall that a \emph{$2$-functor} $F\colon\CC\to\CD$ between ordered categories is a functor preserving the order on hom-sets; that is,
$$f\leq g \implies Ff\leq Fg$$
for all morphisms $f,g\colon X\to Y$ in $\CC$. Let $\CC^{\co}$ denote the ordered category with the same objects and morphisms as in $\CC$, but \emph{reversing} the order on each hom-set of $\CC$. Then it is routine to check that
\begin{align*}
(-)_{\nat}&\colon(\QFOrd)^{\co}\to\QFDist,\quad (f\colon X\to Y)\longmapsto(f_{\nat}\colon X\oto Y)\quad\text{and}\\
(-)^{\nat}&\colon(\QFOrd)^{\op}\to\QFDist,\quad (f\colon X\to Y)\longmapsto(f^{\nat}\colon Y\oto X)
\end{align*}
are both $2$-functors of ordered categories, and furthermore:

\begin{prop} {\rm(See \cite{Heymans2010}.)}  \label{Kan-functor}
Both
\begin{align*}
(-)^*&\colon\QFDist\to(\QFOrd)^{\op},\quad (\phi\colon X\oto Y)\longmapsto(\phi^*\colon\PY\to\PX)\quad\text{and}\\
(-)^{\dag}&\colon\QFDist\to(\QFOrd)^{\co},\quad (\phi\colon X\oto Y)\longmapsto(\phi^{\dag}\colon\PdX\to\PdY)
\end{align*}
are $2$-functors, and one has two pairs of adjoint $2$-functors
$$\bfig
\morphism/@{->}@<4pt>/<800,0>[\QFDist`(\QFOrd)^{\op};(-)^*]
\morphism(800,0)|b|/@{->}@<4pt>/<-800,0>[(\QFOrd)^{\op}`\QFDist;(-)^{\nat}]
\place(350,5)[\mbox{\footnotesize{$\bot$}}]
\efig\quad\text{and}\quad
\bfig
\morphism/@{->}@<4pt>/<800,0>[(\QFOrd)^{\co}`\QFDist.;(-)_{\nat}]
\morphism(800,0)|b|/@{->}@<4pt>/<-800,0>[\QFDist.`(\QFOrd)^{\co};(-)^{\dag}]
\place(420,5)[\mbox{\footnotesize{$\bot$}}]
\efig$$
\end{prop}

\begin{proof}
For any $\FQ$-distributor $\phi\colon X\oto Y$, define
\begin{equation} \label{tphi-hphi-def}
\tphi\colon Y\to\PX,\quad y\longmapsto\phi(-,y)\quad\text{and}\quad\hphi\colon X\to\PdY,\quad x\longmapsto\phi(x,-).
\end{equation}
It is straightforward to check that the assignments $\phi\longmapsto\tphi$ and $\phi\longmapsto\hphi$ give isomorphisms of hom-sets (isomorphisms of complete lattices, indeed, see Theorem~\ref{dist-sup-iso})
\begin{equation} \label{QFOrd-QFDist}
\QFOrd(Y,\PX)\cong\QFDist(X,Y)\cong(\QFOrd)^{\co}(X,\PdY)
\end{equation}
natural in $X,Y\in\QFOrd$.
\end{proof}

It follows immediately from the isomorphisms \eqref{QFOrd-QFDist} that whenever any one of $\phi\colon X\oto Y$, $\tphi\colon Y\to\PX$ or $\hphi\colon X\to\PdY$ is fixed, then so are the other two. The following identities are easy to verify, but quite useful:

\begin{prop} {\rm(See \cite{Shen2014,Shen2013a}.)} \label{tphi-hphi-Yoneda}
For any $\FQ$-distributor $\phi\colon X\oto Y$,
$$\hphi=\uphi\sy_X=\phi^{\dag}\syd_X\quad\text{and}\quad\tphi=\dphi\syd_Y=\phi^*\sy_Y.$$
\end{prop}

The following proposition is an immediate consequence of the above one:

\begin{prop} \label{phi-uphi}
For any $\FQ$-distributors $\phi,\psi\colon X\oto Y$,
$$\phi=\psi\iff\uphi=\upsi\iff\dphi=\dpsi\iff\phi^*=\psi^*\iff\phi^{\dag}=\psi^{\dag}.$$
\end{prop}

With the above preparations, we are now ready to characterize an arbitrary $\FQ$-Galois connection in terms of {\protect\linebreak}$\FQ$-polarities and (dual) $\FQ$-axialities:

\begin{thm} \label{adj-power}
For $\FQ$-order-preserving maps $f\colon X\to Y$ and $g\colon Y\to X$, the following statements are equivalent{\textup:}
\begin{enumerate}[itemsep=-0.7em,label={\rm(\roman*)}]
\item \label{adj-power:X-Y}
    $\bfig
    \morphism/@{->}@<4pt>/<400,0>[X`Y;f]
    \morphism(400,0)|b|/@{->}@<4pt>/<-400,0>[Y`X;g]
    \place(200,5)[\mbox{\footnotesize{$\bot$}}]
    \efig$ is a $\FQ$-Galois connection.
\item \label{adj-power:PX-PdY}
    $\bfig
    \morphism/@{->}@<4pt>/<500,0>[\PX`\PdY;(f_{\nat})_{\ua}]
    \morphism(500,0)|b|/@{->}@<4pt>/<-500,0>[\PdY`\PX;(g^{\nat})^{\da}]
    \place(240,5)[\mbox{\footnotesize{$\bot$}}]
    \efig$ is a $\FQ$-polarity from $X$ to $Y$.
\item \label{adj-power:PY-PX}
    $\bfig
    \morphism/@{->}@<4pt>/<500,0>[\PY`\PX;(f_{\nat})^*]
    \morphism(500,0)|b|/@{->}@<4pt>/<-500,0>[\PX`\PY;(g^{\nat})_*]
    \place(250,5)[\mbox{\footnotesize{$\bot$}}]
    \efig$ is a $\FQ$-axiality from $Y$ to $X$.
\item \label{adj-power:PdY-PdX}
    $\bfig
    \morphism/@{->}@<4pt>/<500,0>[\PdY`\PdX;(f_{\nat})_{\dag}]
    \morphism(500,0)|b|/@{->}@<4pt>/<-500,0>[\PdX`\PdY;(g^{\nat})^{\dag}]
    \place(250,5)[\mbox{\footnotesize{$\bot$}}]
    \efig$ is a dual $\FQ$-axiality from $Y$ to $X$.
\end{enumerate}
\end{thm}

\begin{proof}
\ref{adj-power:X-Y}$\implies$\ref{adj-power:PX-PdY}, \ref{adj-power:X-Y}$\implies$\ref{adj-power:PY-PX} and \ref{adj-power:X-Y}$\implies$\ref{adj-power:PdY-PdX} are immediate consequences of Propositions~\ref{adjoint-graph} and \ref{Isbell-Kan-def}. Conversely, if $(f_{\nat})_{\ua}\dv(g^{\nat})^{\da}$, then $(f_{\nat})_{\ua}=(g^{\nat})_{\ua}$ since one also has $(g^{\nat})_{\ua}\dv(g^{\nat})^{\da}$, and thus Proposition~\ref{phi-uphi} guarantees $f\dv g$; this proves \ref{adj-power:PX-PdY}$\implies$\ref{adj-power:X-Y}. One could derive \ref{adj-power:PY-PX}$\implies$\ref{adj-power:X-Y} and \ref{adj-power:PdY-PdX}$\implies$\ref{adj-power:X-Y} with similar arguments.
\end{proof}

In particular, the $\FQ$-polarity $\bfig
\morphism/@{->}@<4pt>/<500,0>[\PX`\PdY;(f_{\nat})_{\ua}]
\morphism(500,0)|b|/@{->}@<4pt>/<-500,0>[\PdY`\PX;(g^{\nat})^{\da}]
\place(240,5)[\mbox{\footnotesize{$\bot$}}]
\efig$ may be considered as a lifting of the $\FQ$-Galois connection $f\dv g$, since one could easily verify the commutativity of the diagram
$$\bfig
\square|alra|/@{->}@<4pt>`->`->`@{->}@<4pt>/<800,500>[X`Y`\PX`\PdY;f`\sy_X`\syd_Y`(f_{\nat})_{\ua}]
\square|blrb|/@{<-}@<-4pt>```@{<-}@<-4pt>/<800,500>[X`Y`\PX`\PdY;g```(g^{\nat})^{\da}]
\place(400,505)[\mbox{\footnotesize{$\bot$}}] \place(400,5)[\mbox{\footnotesize{$\bot$}}]
\efig$$
with the identities given in Proposition~\ref{tphi-hphi-Yoneda}.

\subsection{$\FQ$-distributors vs. $\FQ$-polarities and {\textup(}dual{\textup)} $\FQ$-axialities}

The interaction between $\FQ$-distributors, $\FQ$-polarities and (dual) $\FQ$-axialities is much more profound than what is revealed in Proposition~\ref{Isbell-Kan-def}; the aim of this last subsection is to show that there exist bijective correspondences between them. The prototypes of the results below come from quantaloid-enriched categories \cite{Shen2014,Shen2013a}; nevertheless, we will provide their proofs in order-theoretic terms so that prior reading of \cite{Shen2014,Shen2013a} is not required for the purpose of understanding this subsection.

\begin{prop} \label{QFDist-QSup}
Let $X$ and $Y$ be $\FQ$-preordered $\FQ$-subsets.
\begin{enumerate}[label={\rm(\arabic*)}]
\item \label{QFDist-QSup:Isbell}
    Every $\FQ$-polarity from $X$ to $Y$ is of the form $\uphi\dv\dphi$ for some $\FQ$-distributor $\phi\colon X\oto Y$.
\item \label{QFDist-QSup:Kan}
    Every $\FQ$-axiality from $X$ to $Y$ is of the form $\phi^*\dv\phi_*$ for some $\FQ$-distributor $\phi\colon Y\oto X$.
\item \label{QFDist-QSup:dual-Kan}
    Every dual $\FQ$-axiality from $X$ to $Y$ is of the form $\phi_{\dag}\dv\phi^{\dag}$ for some $\FQ$-distributor $\phi\colon Y\oto X$.
\end{enumerate}
\end{prop}

\begin{proof}
\ref{QFDist-QSup:Isbell} Let $\bfig
\morphism/@{->}@<4pt>/<500,0>[\PX`\PdY;f]
\morphism(500,0)|b|/@{->}@<4pt>/<-500,0>[\PdY`\PX;g]
\place(240,5)[\mbox{\footnotesize{$\bot$}}]
\efig$ be a $\FQ$-polarity from $X$ to $Y$, then $\hphi:=(X\to^{\sy_X}\PX\to^f\PdY)$ defines a $\FQ$-distributor $\phi\colon X\oto Y$ (see Proposition~\ref{Kan-functor}). We claim that $f=\uphi$. Indeed, for any $\mu\in\PX$,
\begin{align*}
f\mu&=f(\mu\comp 1_X^{\nat})&(\text{Equation~\eqref{PX-def-eq}})\\
&=f\Big(\bv_{x\in X}\mu(x)\comp 1_X^{\nat}(-,x)\Big)&(\text{Proposition~\ref{QFRel-comp-imp}})\\
&=f\Big(\bv_{x\in X}\mu(x)\comp\sy_X x\Big)&\\
&=\bw_{x\in X}f(\mu(x)\comp\sy_X x)&(\text{Theorem~\ref{la-condition} and Remark~\ref{PdX-order}})\\
&=\bw_{x\in X}f\sy_X x\lda\mu(x)&(\text{Theorem~\ref{la-condition} and Example~\ref{PX-tensor}})\\
&=\bw_{x\in X}\hphi x\lda\mu(x)\\
&=\bw_{x\in X}\phi(x,-)\lda\mu(x)&(\text{Equations~\eqref{tphi-hphi-def}})\\
&=\phi\lda\mu&(\text{Proposition~\ref{QFRel-comp-imp}})\\
&=\uphi\mu,
\end{align*}
as desired.

\ref{QFDist-QSup:Kan} If $\bfig
\morphism/@{->}@<4pt>/<500,0>[\PX`\PY;f]
\morphism(500,0)|b|/@{->}@<4pt>/<-500,0>[\PY`\PX;g]
\place(250,5)[\mbox{\footnotesize{$\bot$}}]
\efig$ is a $\FQ$-axiality from $X$ to $Y$, then $\tphi:=(X\to^{\sy_X}\PX\to^f\PY)$ defines a $\FQ$-distributor $\phi\colon Y\oto X$ (see Proposition~\ref{Kan-functor}) with $f=\phi^*$.

\ref{QFDist-QSup:dual-Kan} If $\bfig
\morphism/@{->}@<4pt>/<500,0>[\PdX`\PdY;f]
\morphism(500,0)|b|/@{->}@<4pt>/<-500,0>[\PdY`\PdX;g]
\place(250,5)[\mbox{\footnotesize{$\bot$}}]
\efig$ is a dual $\FQ$-axiality from $X$ to $Y$, then $\hphi:=(Y\to^{\syd_Y}\PdY\to^g\PdX)$ defines a $\FQ$-distributor $\phi\colon Y\oto X$ (see Proposition~\ref{Kan-functor}) with $g=\phi^{\dag}$.

The verification of the details of \ref{QFDist-QSup:Kan} and \ref{QFDist-QSup:dual-Kan} is similar to \ref{QFDist-QSup:Isbell}, so we leave it to the readers.
\end{proof}

For $X,Y\in\QFOrd$, let us take a closer look at
\begin{enumerate}[label=(\arabic*)]
\item \label{QDist-X-oto-Y}
    $\FQ$-distributors $X\oto Y$, and
\item \label{Q-polarity-X-to-Y}
    $\FQ$-polarities from $X$ to $Y$.
\end{enumerate}
Propositions~\ref{Isbell-Kan-def} and \ref{QFDist-QSup} give us an assignment $\phi\longmapsto(\uphi\dv\dphi)$ from \ref{QDist-X-oto-Y} to \ref{Q-polarity-X-to-Y}, and an assignment ${(f\dv g)\longmapsto\phi}$ with $\hphi=f\sy_X$ from \ref{Q-polarity-X-to-Y} to \ref{QDist-X-oto-Y}. On one hand, in Proposition~\ref{QFDist-QSup} it is already shown that the composition of the two assignments is the identity when starting from \ref{Q-polarity-X-to-Y}. On the other hand, starting from a $\FQ$-distributor $\phi\colon X\oto Y$ one has $\hphi=\uphi\sy_X$ by Proposition~\ref{tphi-hphi-Yoneda}, showing that the composition in the other direction also produces the identity. Therefore, these two assignments are inverse to each other, and thus \ref{QDist-X-oto-Y} and \ref{Q-polarity-X-to-Y} are bijective to each other.

Since the same argument shows that $\FQ$-distributors $X\oto Y$ correspond bijectively to $\FQ$-axialities from $Y$ to $X$ and also dual $\FQ$-axialities from $Y$ to $X$, in conjunction with \eqref{QFOrd-QFDist} we have proved:

\begin{thm} \label{dist-polarity-axiality-bijection}
Let $X$ and $Y$ be $\FQ$-preordered $\FQ$-subsets. Then the following items are bijective to each other{\textup:}
\begin{itemize}
\item $\FQ$-distributors $X\oto Y${\textup;}
\item $\FQ$-order-preserving maps $Y\to\PX${\textup;}
\item $\FQ$-order-preserving maps $X\to\PdY${\textup;}
\item $\FQ$-polarities from $X$ to $Y${\textup;}
\item $\FQ$-axialities from $Y$ to $X${\textup;}
\item dual $\FQ$-axialities from $Y$ to $X$.
\end{itemize}
\end{thm}

The bijections in the above theorem are indeed isomorphisms of hom-sets of the ordered categories concerned in this paper which, moreover, are actually isomorphisms of complete lattices:

\begin{thm} \label{dist-sup-iso}
Let $X$ and $Y$ be $\FQ$-preordered $\FQ$-subsets. There are isomorphisms of complete lattices
\begin{align*}
\QFDist(X,Y)&\cong\QFOrd(Y,\PX)\cong(\QFOrd)^{\co}(X,\PdY)\\
&\cong(\QFSup)^{\co}(\PX,\PdY)\cong\QFSup(\PY,\PX)\cong(\QFSup)^{\co}(\PdY,\PdX).
\end{align*}
\end{thm}

\begin{proof}
We prove $\QFDist(X,Y)\cong(\QFSup)^{\co}(\PX,\PdY)$ as an example, and the rest isomorphisms are similar. To see that the bijections $\phi\longmapsto\uphi$ and $f\longmapsto\phi$ with $\hphi=f\sy_X$ establish isomorphisms of complete lattices, it suffices to show that they are order-preserving. Indeed,
\begin{align*}
\phi\leq\phi'\ \text{in}\ \QFDist(X,Y)&\iff\forall\mu\in\PX\colon\uphi\mu=\phi\lda\mu\leq\phi'\lda\mu=\uphi'\mu\ \text{in}\ \QFDist\\
&\iff\forall\mu\in\PX\colon\uphi\mu\geq\uphi'\mu\text{ in } \PdY\\
&\iff\uphi\leq\uphi'\ \text{in}\ (\QFSup)^{\co}(\PX,\PdY),
\end{align*}
and
\begin{align*}
f\leq f'\ \text{in}\ (\QFSup)^{\co}(\PX,\PdY)&\iff\forall\mu\in\PX\colon f\mu\geq f'\mu\ \text{in}\ \PdY\\
&\iff\forall\mu\in\PX\colon f\mu\leq f'\mu\text{ in }\QFDist\\
&\ \implies\forall x\in X\colon f\sy_X x\leq f'\sy_X x\text{ in }\QFDist\\
&\iff\forall x\in X\colon \phi(x,-)\leq\phi'(x,-)\text{ in }\QFDist\\
&\iff\phi\leq\phi'\text{ in }\QFDist(X,Y),
\end{align*}
where $\phi$ and $\phi'$ are determined by $\hphi=f\sy_X$ and $\widehat{\phi'}=f'\sy_X$.
\end{proof}

\section*{Acknowledgements}

The first named author acknowledges the support of the grants MTM2015-63608-P (MINECO/FEDER,~UE) and \text{IT974-16} (Basque Government). The second named author acknowledges the support of National Natural Science Foundation of China (No.~11771310). The third named author acknowledges the support of National Natural Science Foundation of China (No.~11701396) and the Fundamental Research Funds for the Central Universities (No.~YJ201644).

This work was initiated while the second and the third named authors were visiting Department of Mathematics at University of the Basque Country UPV/EHU from late November to early December in 2016, with the kind support and hospitality of the first named author and Iraide Mardones-P\'{e}rez.





\end{document}